\documentclass[aps,pra,twocolumn,floats, showpacs,superscriptaddress,groupedaddress, 10pt]{revtex4-1}  % for review and submission
\usepackage{amsfonts}
\usepackage{amsmath,amssymb,bbm}
\usepackage{amsthm}
\usepackage{pspicture}
\usepackage{epsfig}
\usepackage{calc}
\usepackage{times}
\usepackage{color}
\usepackage{graphicx} 
\usepackage{enumerate}  
\usepackage{hyperref}

\usepackage{xcolor}

\usepackage{natbib}
\newcommand{\assign}{:=}
\usepackage{bbm}

\newtheorem{hypothesis}{Hypothesis}

\newcommand{\cA}{{\mathcal A}}
\newcommand{\cE}{{\mathcal E}}
\newcommand{\cH}{{\mathcal H}}

\newcommand{\cN}{{\mathcal N}}
\newcommand{\cR}{{\mathcal R}}

\newtheorem{theorem}{Theorem}
\newtheorem{definition}{Definition}
\newtheorem{lemma}{Lemma}
\newtheorem{corollary}{Corollary}
\theoremstyle{definition}

\begin{document}

\title{Code properties from holographic geometries}
	
	\author{Fernando \surname{Pastawski}}
	\affiliation{Dahlem Center for Complex Quantum Systems, Freie Universit{\"a}t Berlin, 14195 Berlin, Germany 
	}
	\author{John \surname{Preskill}}
	\affiliation{Institute for Quantum Information and Matter, California Institute of Technology, Pasadena, California 91125, USA
	}	

\date{\today}

\begin{abstract}
Almheiri, Dong, and Harlow \cite{Almheiri2015} proposed a highly illuminating connection between the AdS/CFT holographic correspondence and operator algebra quantum error correction (OAQEC). Here we explore this connection further. We derive some general results about OAQEC, as well as results that apply specifically to quantum codes which admit a holographic interpretation. We introduce a new quantity called {\it price}, which characterizes the support of a protected logical system, and find constraints on the price and the distance for logical subalgebras of quantum codes. We show that holographic codes defined on bulk manifolds with asymptotically negative curvature exhibit {\it uberholography}, meaning that a bulk logical algebra can be supported on a boundary region with a fractal structure. We argue that, for holographic codes defined on bulk manifolds with asymptotically flat or positive curvature, the boundary physics must be highly nonlocal, an observation with  potential implications for black holes and for quantum gravity in AdS space at distance scales small compared to the AdS curvature radius. 
\end{abstract}

\maketitle

\section{Introduction}

Quantum error correction and the holographic principle are two of the most far-reaching ideas in contemporary physics. Quantum error correction provides a basis for believing that scalable quantum computers can be built and operated in the foreseeable future. The AdS/CFT holographic correspondence is currently our best tool for understanding nonperturbative quantum gravity. In a remarkable paper \cite{Almheiri2015}, Almheiri, Dong, and Harlow suggested that these two deep ideas are closely related.

The AdS/CFT correspondence is an exact duality between two quantum theories --- quantum gravity in $(D+1)$-dimensional anti-de Sitter space and a conformally invariant quantum field theory (without gravity) defined on its $D$-dimensional boundary. The observables of the two theories are related by a complex dictionary, which maps local operators supported deep inside the bulk spacetime to highly nonlocal operators acting on the boundary CFT. Almheiri {\it et al.} proposed interpreting this dictionary as the encoding map of a quantum error-correcting code, where the code subspace is the low-energy sector of the CFT. Bulk local operators are regarded as ``logical'' operators which map the code subspace $\cH_C$ to itself, and are well protected against erasure of portions of the boundary.  The holographic dictionary is an encoding map which embeds the logical system inside the physical Hilbert space $\cH$ of the CFT. This proposal provides a rich and enticing new perspective on the relationship between the emergent bulk geometry and the entanglement structure of the CFT. 

To model holography faithfully, the quantum error-correcting code must have special properties which invite a geometrical interpretation. Code constructions which realize the ideas in \cite{Almheiri2015}, based on tensor networks which cover the associated bulk geometry, were constructed in \cite{Pastawski2015} and extended in \cite{Hayden2016}. Importantly, it was shown \cite{Hayden2016} that codes can have holographic properties even when the underlying bulk geometry does not have negative curvature; this insight may broaden our perspective on how AdS space is special.
 
Our goal in this paper is to develop these ideas further. 
Our motivation is twofold. On one hand, holographic codes have opened a new avenue in quantum coding theory, and it is worthwhile to explore more deeply how geometric insights can provide new methods for deriving code properties. 
On the other hand, holographic codes provide a useful tool for sharpening the connections between holographic duality and quantum information theory. Specifically, as emphasized in \cite{Almheiri2015}, holographic codes are best described and analyzed using the language of operator algebra quantum error correction \cite{Kribs2005, Beny2007, Beny2007a, Nielsen2007}. This powerful framework deserves to be better known, and much of this paper will be devoted to amplifying and applying it. 

We view our work here as a step along the road toward answering a fundamental question about quantum gravity and holography: What is the bulk? The AdS/CFT correspondence has bestowed many blessings, but should be regarded as a crutch that must eventually be discarded to clear the way for future progress in quantum cosmology. We think that strengthening the ties between geometric and algebraic properties will be empowering, and that operator algebra quantum error correction can help to point the way.  As examples, we provide an algebraic characterization of a point in the bulk spacetime, and discuss criteria for local correctability of the boundary theory. We also elaborate on the notion of {\it uberholography}, in which bulk physics can be reconstructed on a boundary subsystem with fractal geometry.

\subsection{Outline}

In \S \ref{sec:distance_defs} we review the formalism of operator algebra quantum error correction (OAQEC). We explain how the notion of code distance can be applied to a subalgebra of a quantum code's logical algebra, and we introduce the complementary notion of the {\it price} of a logical subalgebra, the size of the minimal subsystem of the physical Hilbert space which supports the logical subalgebra. We also derive some inequalities relating distance and price, and note that distance and price are equal for the logical subalgebra supported on a bulk point.  In \S \ref{sec:Connection} we review the connection between holography and quantum error correction, emphasizing the role of OAQEC in the analysis of holographic codes. We formulate the {\it entanglement wedge hypothesis}, a geometric criterion that determines whether a bulk logical subalgebra can be reconstructed on a specified boundary region, and work out some of its implications. We also discuss properties of punctures in the bulk geometry, which provide a crude description of black holes inside the bulk. 

In \S \ref{sec:uberholography} we explain the idea of uberholography, and compute the universal fractal dimension which determines how price and distance scale with system size for a holographic code defined on a hyperbolic disk. In \S \ref{sec:quantum-Markov} we investigate the conditions for {\it local correctability} in a holographic code, where by ``local'' we mean that the erasure of a small connected boundary region $R$ can be corrected by a recovery map which acts only in a slightly larger region containing $R$. We explain that holographic codes are locally correctable when the bulk geometry is negatively curved asymptotically, but not for asymptotic flat or positive curvature. We interpret this property as a signal of nonlocal physics on the boundary in the flat and positively curved cases, and we also relate properties of black holes to features of holographic codes with positive curvature.  In \S \ref{sec:holographic-singleton} we use geometrical and entropic arguments to prove a {\it strong quantum Singleton bound} for holographic codes, which constrains the price and distance of a logical subalgebra. \S \ref{sec:DiscussionANdOutlook} contains some concluding comments.

\section{Operator algebra quantum error correction}\label{sec:distance_defs}
In this section we briefly review the principles of operator algebra quantum error correction (OAQEC) \cite{Kribs2005, Beny2007, Beny2007a, Nielsen2007}, providing a foundation for the discussion of holographic codes. We explain how the notion of code distance can be generalized to the OAQEC setting. We also introduce a related but complementary notion, the {\it price} of a code and of a logical operator algebra. In a holographic context, the distance of a bulk logical algebra characterizes how well the bulk degrees of freedom are protected against erasure of portions of the boundary, while its price characterizes the minimal boundary region on which the bulk degrees of freedom can be reconstructed. 

\subsection{Von Neumann algebras}\label{subsec:VNalgebra}

Since we will formulate quantum error correction in an operator algebra framework, we begin by reviewing the structure of finite-dimensional Von Neumann algebras. For a finite-dimensional complex Hilbert space ${\cal H}$, a Von Neumann algebra on ${\cal H}$ is a complex vector space of linear operators acting on ${\cal H}$ which is closed under multiplication and Hermitian adjoint. Any such algebra ${\cal A}$ can be characterized in the following way. The Hilbert space ${\cal H}$ contains a subspace with a direct sum decomposition, such that each summand is a product of two tensor factors:
\begin{align}\label{eq:H-decompose}
{\cal H} \supseteq \bigoplus_\alpha {\cal H}_\alpha \otimes {\cal H}_{\bar \alpha},
\end{align}
where ${\cal H}_\alpha$ has dimension $d_\alpha$ and ${\cal H}_{\bar \alpha}$ has dimension $d_{\bar \alpha}$. 
The Von Neumann algebra ${\cal A}$ can be expressed as 
\begin{align}\label{eq:A-decompose}
{\cal A} = \bigoplus_\alpha {\cal M}_\alpha \otimes I_{\bar\alpha},
\end{align}
where ${\cal M}_\alpha$ denotes the algebra of $d_\alpha\times d_\alpha$ matrices and $I_{\bar\alpha}$ denotes the $d_{\bar \alpha}\times d_{\bar \alpha}$ identity matrix. The {\it commutant} ${\cal A}'$ of ${\cal A}$ contains all operators on ${\cal H}$ which commute with all operators in ${\cal A}$, and can be expressed as 
\begin{align}\label{eq:commutant}
{\cal A}' = \bigoplus_\alpha I_\alpha \otimes {\cal M}_{\bar\alpha},
\end{align}
The center $Z({\cal A})$ of ${\cal A}$, which is also the center of its commutant, contains all elements of the form  
\begin{align}\label{eq:center-element}
\bigoplus_\alpha m_\alpha ~ I_\alpha \otimes I_{\bar\alpha};
\end{align}
note that the center is abelian. 

A nontrivial Von Neumann algebra (with more than one summand) describes a quantum system with superselection sectors. 
We may regard $\alpha$ as a label which specifies a sector with a specified value of a locally conserved charge. 
% Coherent superpositions of different values of $\alpha$ are unphysical
By focusing on $\cA$, we confine attention to operators that preserve $\alpha$. 
To interpret the decomposition eq.(\ref{eq:H-decompose}), we imagine a system shared by two parties, Alice and Bob, where in each $\alpha$-sector the parties have equal and opposite charges. The algebras ${\cal A}$ and ${\cal A}'$ capture the charge-preserving operations that can be applied by Alice and Bob respectively. Equivalently, we may say that a nontrivial Von Neumann algebra describes a system which encodes both classical and quantum information, where operators in the center $Z({\cal A})= Z({\cal A}')$ act only on the classical data (the label $\alpha$), while ${\cal M}_\alpha$ acts on the quantum data in the sector labeled by $\alpha$. 

In operator algebra quantum error correction (OAQEC), we consider ${\cal H}_C$ to be a code subspace of a larger physical Hilbert space; hence ${\cal A}$ and ${\cal A}'$ are algebras of {\it logical } operators which preserve the code subspace. In the case where there is a single summand and ${\cal M}_{\bar \alpha}$ is one-dimensional, ${\cal A}$ is the complete algebra of logical operators. This is the standard setting of quantum error-correcting codes. If there is a single summand and ${\cal M}_{\bar\alpha}$ is nontrivial, then ${\cal A}$ is the algebra of ``bare'' logical operators in a {\it subsystem code}. In this setting, the code subspace 
\begin{align}
{\cal H}_C = {\cal H}_\alpha\otimes {\cal H}_{\bar\alpha}
\end{align}
has a decomposition into a protected tensor factor ${\cal H}_{\alpha}$ and a ``gauge'' factor ${\cal H}_{\bar\alpha}$, and ${\cal A}$ acts only on the protected system. 

The more general setting, with a nontrivial sum over $\alpha$, arises naturally in the context of holographic duality, where the code subspace corresponds to the low-energy sector of a conformal field theory whose gravitational dual is a bulk system with emergent gauge symmetry.   
The abelian center $Z({\cal A})$ of ${\cal A}$ can, for example, encode classical data of the bulk geometry (see Ref. \cite{Donnelly2016} for a recent tensor network interpretation). 
An important example of such a classical variable contained in ${\cal A}$ is the area operator (see Ref. \cite{Almheiri2016}) which arises in the Ryu-Takayanagi formula relating boundary entropy to bulk geometry. 

Another reason the OAQEC formalism is convenient in discussions of holography is that we can formulate the notion of {\it complementary recovery} \cite{Harlow2016} using this language.
If the physical (boundary) Hilbert space has a decomposition as a product $R R^c$ of two subsystems, we may ask whether a subalgebra ${\cal A}$ acting on the code space can be ``reconstructed'' as an algebra of physical operators with support on $R$, in which case we may say that erasure of $R^c$ can be corrected for the algebra ${\cal A}$. 
We say that the code exhibits complementary recovery if the logical subalgebra ${\cal A}$ can be reconstructed on $R$ and its commutant ${\cal A}'$ can be reconstructed on $R^c$. Equivalently, complementary recovery means that erasure of the physical subsystem $R^c$ is correctable with respect to ${\cal A}$ and erasure of the complementary physical subsystem $R$ is correctable with respect to ${\cal A}'$.

\subsection{Correctability}

Quantum error correction is a way of protecting properly encoded quantum states from the potentially damaging effects of noise with suitable properties. The noise can be described by a completely positive trace-preserving map (CPTP map), also called a quantum channel. A channel is a linear map which takes density operators to density operators; saying that the channel is ``completely'' positive means that the positivity of the density operator is preserved even when the channel acts on a system which is entangled with other systems. 

A channel ${\cal N}$ has an operator sum representation (also called a Kraus representation) of the form
\begin{align}\label{eq:kraus-schroedinger}
{\cal N}(\rho) = \sum_a N_a\rho N_a^\dagger,
\end{align}
where the condition
\begin{align}\label{eq:kraus-normalize}
\sum_a N_a^\dagger N_a = I
\end{align}
ensures that ${\rm tr}[N(\rho)] = {\rm tr}[\rho]$. The operators $\{N_a\}$ appearing in the Kraus representation are called Kraus operators. 
If there is only one Kraus operator in the sum, then the map is unitary, taking pure states to pure states. If there are two or more linearly independent Kraus operators, the map ${\cal N}$ describes a decoherence process, in which pure states can evolve to mixed states. 

Eq.(\ref{eq:kraus-schroedinger}) is the Schr\"{o}dinger picture description of the channel, in which ${\cal N}$ maps states to states. Since we will be particularly interested in whether {\it operators} (rather than states) are well protected against noise, we will find it more convenient to consider the Heisenberg picture description in which states are fixed and operators evolve. In this picture, the noise acts on the operator $X$ according to 
\begin{align}\label{eq:kraus-heisenberg}
{\cal N}^\dagger(X) = \sum_a N_a^\dagger X N_a.
\end{align}
We say that ${\cal N}^\dagger$ is the dual map of ${\cal N}$, also called the adjoint map of ${\cal N}$. The condition eq.(\ref{eq:kraus-normalize}) ensures that ${\cal N}^\dagger$ maps the identity operator to itself. 

We will consider a quantum system with Hilbert space ${\cal H}$, and a noise channel ${\cal N}$ acting on the system. Quantum error correction is a process that reverses the effect of ${\cal N}$. This error correction process is itself a channel, called the recovery channel and denoted ${\cal R}$. Unless ${\cal N}$ is unitary, error correction is not possible for arbitrary states of the system. Instead we consider a subspace ${\cal H}_C$ of ${\cal H}$, which is called a quantum error-correcting code (QECC), and settle for a recovery channel that corrects ${\cal N}$ acting on states of ${\cal H}_C$.  We say that ${\cal R}$ corrects ${\cal N}$ on code subspace ${\cal H}_C$ if, for any density operator $\rho$ supported on ${\cal H}_C$, 
\begin{align}
\left({\cal R}\circ {\cal N}\right)(\rho) = \rho,
\end{align}
and we say that the noise channel ${\cal N}$ is correctable on ${\cal H}_C$ if there exists a recovery operator ${\cal R}$ which corrects ${\cal N}$.

In the Heisenberg picture language, we may consider an algebra of logical operators which act on the code space. We denote the set of linear operators mapping ${\cal H}$ to ${\cal H}$ by ${\cal L}({\cal H})$.  If $P$ denotes the orthogonal projector from ${\cal H}$ to ${\cal H}_C$, then an operator $X\in {\cal L}({\cal H})$ is {\it logical} if $[X,P] = 0$; hence $X$ maps ${\cal H}_C = P{\cal H}$ to itself:
\begin{align}
X{\cal H}_C = XP{\cal H} = PX{\cal H} \subseteq {\cal H}_C.
\end{align}
It is clear from this definition that if $X$ is logical, so is its Hermitian adjoint $X^\dagger$ (because $P=P^\dagger$); furthermore, a linear combination of logical operators is logical and so is a product of logical operators. Hence the set of all logical operators forms an algebra, which we may call the complete logical operator of the code. The theory of operator algebra quantum error correction addresses whether a subalgebra of this complete logical algebra can be protected against noise. 

Sometimes we are only interested in how a logical operator $X$ acts on the code space, and so consider the corresponding operator $PXP$, which has support only on ${\cal H}_C$. Operators of this type are also closed under multiplication, because if $X$ and $Y$ both commute with $P$, then
\begin{align}
PXP\cdot PYP = P(XY)P. 
\end{align}
That is, if ${\cal A}$ is an algebra of logical operators in ${\cal L}({\cal H})$, then $P{\cal A}P$ is an algebra of logical operators in ${\cal L}({\cal H}_C)$. Two logical operators $X$ and $\tilde{X}$ might differ as elements of ${\cal L}({\cal H})$ yet act on the code space in the same way because $PXP = P\tilde{X}P$. In that case we say that $X$ and $\tilde{X}$ are {\it logically equivalent}, denoted  $X \sim_P \tilde{X}$.

Now we can formulate the notion of error correction in the Heisenberg picture. 
\begin{definition}[correctability]\label{Def:CorrectableNoise map} The noise channel ${\cal N}$ is correctable on the code space $\cH_C = P \cH$ with respect to the operator $X\in {\cal L}({\cal H})$ if and only if there exists a recovery channel $\cR$ such that
\begin{align}\label{eq:Correctability}
P(\cR \circ \cN)^\dagger(X) P = P X P.
\end{align}
\end{definition}
This means that the operator $X$ and the recovered operator $(\cR \circ \cN)^\dagger(X)$ act on the code space in the same way, though they may act differently on state vectors outside the code space. Because this condition is linear in $X$, the operators with respect to which ${\cal N}$ is correctable form a linear space. 

In an important series of works \cite{Poulin2005, Kribs2005, Kribs2006}, culminating in \cite{Beny2007a, Beny2007} the necessary and sufficient conditions for correctability of a logical algebra were derived. 
\begin{theorem}[criterion for correctability]\label{thm:correctability} Given code subspace $\cH_C = P\cH$ and logical subalgebra ${\cal A}$, 
the noise channel ${\cal N}$ with Kraus operators $\{N_a\}$ is correctable with respect to ${\cal A}$ if and only if
\begin{align}\label{eq:beny-condition}
[P N_a^\dagger N_bP , X] = 0
\end{align} for all $X\in \cA$ and each pair of Kraus operators $N_a, N_b$.
\end{theorem}
If ${\cal A}$ is the code's complete logical algebra, eq.(\ref{eq:beny-condition}) becomes
\begin{align}
PN_a^\dagger N_b P= c_{ab}P,
\end{align}
which is the well known Knill-Laflamme error correction condition \cite{Knill1997}. 
More generally, eq.(\ref{eq:beny-condition}) says that $PN_a^\dagger N_b P$ lies in the commutant ${\cal A}'$ of ${\cal A}$. 
Invoking the general structure of Von Neumann algebras reviewed in \ref{subsec:VNalgebra}, we see from eq.(\ref{eq:commutant}) that $PN_a^\dagger N_b P$ is supported on the second factor of each summand. In effect, eq.(\ref{eq:beny-condition}) means that the Knill-Laflamme condition is satisfied in each superselection sector of the logical algebra. 

\subsection{Erasure and reconstruction}

A noise channel of particular interest is the {\it erasure} channel.  To define the erasure channel, we consider a decomposition of the Hilbert space ${\cal H}$ as a product of two tensor factors,
\begin{align}
{\cal H} = {\cal H}_R \otimes {\cal H}_{R^c};
\end{align}
we'll sometimes express this decomposition more succinctly as $R R^c$. Anticipating the geometrical interpretation of holographic codes, we will call $R$ a {\it region} and say that $R^c$ is its complementary region. We say that $R$ is {\it erased} when the quantum information in $R$ is lost while the information in $R^c$ is retained. A noise channel describing this process is
\begin{equation}\label{eq:erasure}
\Delta_R(\rho) = \sigma_R \otimes {\rm tr}_R(\rho),
\end{equation}
which is called the {\it erasure map} on $R$, or the {\it depolarizing map} on $R$; it throws away the state of $R$ and replaces it by the fixed state $\sigma_R$.

As for any noise channel, we say that the erasure channel ${\cal N}=\Delta_R$ is correctable with respect to the operator $X$ if there is a recovery operator ${\cal R}$ satisfying eq.(\ref{eq:Correctability}). As a convenient shorthand, we say that the subsystem $R$ is correctable if erasure of $R$ is correctable:
\begin{definition}[correctable subsystem]
Given a code subspace ${\cal H}_C=P{\cal H}$ and a logical subalgebra ${\cal A}$, a subsystem $R$ of ${\cal H}$ is {\it correctable} with respect to ${\cal A}$ if erasure of $R$ is a correctable map with respect to ${\cal A}$.
\end{definition}
Whether $R$ is correctable does not depend on how we choose the state $\sigma_R$ in eq.(\ref{eq:erasure}); if ${\cal R}$ recovers from $\Delta_R$, then ${\cal R}\circ \Delta_R$ recovers from $\Delta_R'$, where $\Delta_R'(\rho) = \sigma_R'\otimes {\rm tr}_R(\rho)$.

For the special case of erasure the criterion for correctability in Theorem \ref{thm:correctability} simplifies.  
We may choose $\sigma_R \propto I_R$, in which case the Kraus operators realizing $\Delta_R$ may be chosen to be (up to normalization) the complete set of Pauli operators supported on $R$, which constitute a complete basis for operators acting on $R$.
More generally, we may realize $\Delta_R$ by taking the Kraus operators in eq. \eqref{eq:kraus-schroedinger} as a Haar average over the unitary operators supported on $R$.
We conclude:

\begin{lemma} [criterion for correctability of a subsystem]\label{lemma:subsys-correct}
Give code subspace $\cH_C = P\cH$ and logical subalgebra $\cA$, subsystem $R$ of ${\cal H}$ is correctable with respect to $\cA$ if and only if
\begin{equation}\label{eq:erasure-criterion}[PYP,X] = 0\end{equation}
for all $X\in \cA$ and every operator $Y$ supported on $R$.
\end{lemma}
Thus erasure of $R$ is correctable with respect to a logical algebra ${\cal A}$ if and only if $PYP$ lies in the commutant ${\cal A}'$ of ${\cal A}$ for any operator $Y$ supported on $R$. Because $X$ is logical ($[X,P]=0$), this criterion can also be written as $P[Y,X]P=0$; that is, the commutator $[Y,X]$ maps the code space to its orthogonal complement.  If ${\cal A}$ is the code's complete logical algebra, the criterion for correctability of erasure becomes
\begin{align}
PYP = c P,
\end{align}
the Knill-Laflamme criterion for erasure correction \cite{Knill1997}. 

If erasure of $R$ is correctable with respect to logical operator $X$, then it is possible to find an operator $\tilde{X}$ which is logically equivalent to $X$ ($PXP=P\tilde{X}P$) such that $\tilde{X}$ is supported on the complementary subsystem $R^c$. 
Borrowing the language of the AdS/CFT correspondence, we may say that $X$ can be ``reconstructed'' on $R^c$. In the quantum information literature, one says that $X$ can be ``cleaned'' on $R$, meaning that there is an equivalent logical operator that acts trivially on the correctable set.

To see why this reconstruction is possible, we may consider the dual $\Delta_R^\dagger$ of the erasure map $\Delta_R$, which satisfies
\begin{equation}
{\rm tr}\left(X\Delta_R(\rho)\right)= {\rm tr}\left(\Delta_R^\dagger(X)\rho\right).
\end{equation}
By the definition of correctability, eq.(\ref{eq:Correctability}), if $R$ is correctable with respect to $X$ then there is a recovery map ${\cal R}_R$ which corrects erasure, such that 
\begin{align}
P\left({\cal R}_R\circ \Delta_R\right)^\dagger (X) P = P\left(\Delta_R^\dagger\circ {\cal R}_R^\dagger\right)(X)P = PXP.
\end{align}
Furthermore, the dual map $\Delta_R^\dagger$ takes any (not necessarily logical) operator $Y$ to an operator which acts trivially on  $R$:
\begin{equation}\label{eq:dagger-map-identity}
\Delta_R^\dagger(Y) = I_R \otimes \tilde Y_{R^c}.
\end{equation}
We see that $\left(\Delta_R^\dagger\circ {\cal R}_R^\dagger\right)(X)$ is logically equivalent to $X$, and supported on $R^c$; that is, it is a reconstruction of $X$ on the complement of the erased subsystem.

To understand eq.(\ref{eq:dagger-map-identity}) we argue as follows. Consider a unitary map supported on $R$, under which
\begin{align}
\rho \mapsto \rho' =\left(U_R\otimes I_{R^c}\right) \rho \left(U_R^\dagger\otimes I_{R^c}\right) ;
\end{align}
hence $\rho_{R^c} = \rho'_{R^c}$, and therefore $\Delta_R(\rho) = \Delta_R(\rho')$, from which we infer that
\begin{align}\label{eq:Delta-U-invariant}
&{\rm tr}\left(\Delta_R^\dagger(Y) \rho \right) ={\rm tr}\left(\Delta_R^\dagger(Y) \rho' \right)\notag\\
&= {\rm tr}\left(\left(U_R^\dagger \otimes I_{R^c}\right)\Delta_R^\dagger(Y) \left(U_R\otimes I_{R^c}\right)\rho \right)
\end{align}
If it holds for any state $\rho$, eq.(\ref{eq:Delta-U-invariant}) implies
\begin{align}
\Delta_R^\dagger(Y) =\left( U_R^\dagger \otimes I_{R^c}\right)\Delta_R^\dagger(Y) \left(U_R\otimes I_{R^c}\right)
\end{align}
for any unitary $U_R$. Eq.(\ref{eq:dagger-map-identity}) then follows. Thus we have shown:
\begin{lemma}[reconstruction] \label{lemma:reconstruction}
Give code subspace $\cH_C = P\cH$ and logical subalgebra $\cA$, if subsystem $R$ of ${\cal H}$ is correctable with respect to $\cA$, then ${\cA}$ can be reconstructed on the complementary subsystem $R^c$. That is, for each logical operator in $\cA$, there is a logically equivalent operator supported on $R^c$.
\end{lemma}

\subsection{Distance and price}\label{subsec:distance-price}

In the standard theory of quantum error correction, we consider the physical Hilbert space ${\cal H}$ to have a natural decomposition as a tensor product of small subsystems, for example a decomposition into $n$ qubits (two-level systems); $n$ is called the {\it length} of the code. This decomposition is ``natural'' in the sense of being motivated by the underlying physics --- {\it e.g.}, each qubit might be carried by a separate particle, where the particles interact pairwise. Typically we suppose that the code subspace ${\cal H}_C$ also has a decomposition into ``logical'' qubits; that is, that the dimension of the code space is $2^k$ where $k$ is a positive integer. We may define the {\it distance} $d$ of the code as the size of the smallest set $R$ of physical qubits for which erasure of $R$ is not correctable. Equivalently, $d$ is the size of the smallest region which supports observables capable of distinguishing among distinct logical states. We use the notation $[[n,k,d]]$ for a code with $n$ physical qubits, $k$ logical qubits, and distance $d$. For a given $n$, it is desirable for $k$ and $d$ to be as large as possible, but there is a tradeoff; larger $k$ means smaller $d$ and vice versa. This standard theory can be generalized in some obvious ways; for example, the dimension of the code subspace might not be a power of $2$, or the physical Hilbert space might be decomposed into higher-dimensional subsystems rather than qubits. 

The distance $d$ loosely characterizes the error-correcting power of the code. But if some encoded degrees of freedom are better protected than others, then a more refined characterization can be useful, since the distance captures only the worst case. In holographic codes in particular, bulk degrees of freedom far from the boundary are better protected than bulk degrees of freedom near the boundary. To describe the performance of a holographic code more completely, we may assign a distance value to each of the code's logical subalgebras.

As in the standard theory, we assume the physical Hilbert space is uniformly factorizable, ${\cal H}={\cal H}_0^{\otimes n}$, where ${\cal H}_0$ is finite-dimensional. In applications to quantum field theory, then, ${\cal H}$ is the Hilbert space of a suitably regulated theory; for example, if the theory is defined on a spatial lattice, a subsystem with Hilbert space ${\cal H}_0$ resides at each lattice site. Guided by this picture, we refer to the elementary subsystem as a ``site.'' By a ``region'' $R$ we mean a subset of the $n$ sites, and the number of sites it contains is called the size of $R$, denoted $|R|$.
We may now define the distance of a logical algebra ${\cal A}$.
\begin{definition}[distance]\label{Def:Distance} Given code subspace ${\cal H}_C = P{\cal H}$ and logical subalgebra ${\cal A}$, the {\it distance} $d({\cal A})$ is the size of the smallest region $R$ which is not correctable with respect to ${\cal A}$.
\end{definition}
If ${\cal A}$ is the code's complete logical algebra, then $d({\cal A})$ coincides with the standard definition of distance for a subspace code. In the case of a subsystem code, if ${\cal A}$ is the algebra of ``bare'' logical operators, which act nontrivially on the protected subsystem and trivially on the gauge subsystem, $d({\cal A})$ is the size of the smallest region which supports a nontrivial ``dressed'' logical operator, one which acts nontrivially on the protected subsystem and might act on the gauge subsystem as well. In that case, $d({\cal A})$ coincides with the standard definition of distance for a subsystem code. More generally, we might want to consider multiple ways of decomposing ${\cal H}_C$ into a protected subsystem and its complement, and our definition assigns a distance to each of these protected subsystems. 

For a given code ${\cal H}_C$ and logical algebra ${\cal A}$, we may also consider the smallest possible region $R$ such that {\it all} operators in ${\cal A}$ are supported on $R$. We call the size of this region the {\it price} of the algebra.
\begin{definition} [price] Given code subspace ${\cal H}_C = P{\cal H}$ and logical subalgebra ${\cal A}$,
the \emph{price} $p({\cal A})$ is the size of the smallest region $R$ such that, for every operator $X\in {\cal A}$, there is a logically equivalent operator $\tilde{X}$ which is supported on $R$. 
\end{definition}
As already noted, if region $R$ is correctable with respect to operator $X$, then an operator logically equivalent to $X$ can be reconstructed on the complementary region $R^c$. In this sense the notions of distance and price are dual to one another. The relation between distance and price can be formulated more precisely with some simple lemmas. 
\begin{lemma}[complementarity] \label{lemma:complementarity} Given code subspace ${\cal H}_C=P{\cal H}$ and logical subalgebra ${\cal A}$, where ${\cal H}$ contains $n$ sites, the distance and price of ${\cal A}$ obey
\begin{align}\label{eq:p+d}
p(\cA) + d(\cA) \leq n + 1.
\end{align}
\end{lemma}
\begin{proof}
Consider a region $R$ which is correctable with respect to ${\cal A}$ and also {\it unextendable}, meaning $R$ has the property that adding any additional site makes it noncorrectable. Then there are noncorrectable sets with $|R|+1$ sites, and therefore $d({\cal A}) \le |R|+1$. Furthermore, since $R$ is correctable, all operators in ${\cal A}$ can be reconstructed on its complement $R^c$; hence the $p({\cal A})\le |R^c| = n - |R|$. Adding these two inequalities yields eq.(\ref{eq:p+d}). 
\end{proof}

We may anticipate that if a region $R$ supports a nontrivial logical algebra, then erasing $R$ inflicts an irreversible logical error. This intuition is correct if the algebra is non-abelian. Let us say that a logical subalgebra is non-abelian if it contains two logical operators $X$ and $Y$ such that $PXP$ and $PYP$ are non-commuting. Then we have:
\begin{lemma}[no free lunch]\label{Lemma:NoFreeLunch} Given code subspace ${\cal H}_C=P{\cal H}$ and non-abelian logical subalgebra ${\cal A}$, the distance and price of ${\cal A}$ obey
\begin{align}\label{eq:free-lunch}
d(\cA) \leq p(\cA).
\end{align} 
\end{lemma}
\begin{proof} 
Consider two logical operators $X$ and $Y$ in ${\cal A}$ (both commuting with $P$), such that $PXP$ and $PYP$ are non-commuting. By the definition of $p({\cal A})$, there is a region $R$ with $|R|=p({\cal A})$ such that an operator $\tilde{Y}$ logically equivalent to $Y$ is supported on $R$; hence  
\begin{align}
0 \ne [PYP,PXP] = [P\tilde{Y}P,PXP]= [P\tilde{Y}P,X]. 
\end{align}
This means that region $R$ does not satisfy the criterion for correctability in Lemma \ref{lemma:subsys-correct}, and therefore is not correctable with respect to ${\cal A}$. By the definition of distance, $d(\cA) \le |R| = p(\cA)$, and eq.(\ref{eq:free-lunch}) follows. 
\end{proof}

If $\cA$ is abelian, then eq.(\ref{eq:free-lunch}) need not apply. Consider for example the three-qubit quantum repetition code, spanned by the states $|000\rangle$ and $|111\rangle$, and the logical algebra generated by
\begin{equation}
\bar Z = |000\rangle\langle 000| - |111\rangle\langle 111|.
\end{equation}
This algebra has price $p=1$, because the operator $Z\otimes I\otimes I$, supported on only the first qubit, is logically equivalent of $\bar Z$. On the other hand, the distance is $d=3$; because the logical algebra can be supported on any one of the three physical qubits, it is protected against the erasure of any two qubits. Note that $p+d = 4$, saturating eq.(\ref{eq:p+d}).

For a traditional subspace code, we may define the price of the code as the price of its complete logical algebra, just as we define the code's distance to be the distance of its complete logical algebra. The price and distance of a code are constrained by an inequality which can be derived from the subadditivity of Von Neumann entropy. This constraint on price is a corollary to the following theorem.
\begin{theorem}[constraint on correctable regions]
Consider a code subspace  $\cH_C= P\cH$, where $\cH$ contains $n$ sites, and let  $k = \log \dim \cH_C/ \log \dim \cH_0$. Suppose that $R_1$ and $R_2$ are two disjoint correctable regions. Then
\begin{equation}\label{eq:disjoint-correctable}
n - k \ge |R_1| + |R_2|.
\end{equation}
\end{theorem}
\begin{proof}

Let $A$ denote the code block $\cH_0^{\otimes n}$, let $T$ denote a reference system, and let $|\Phi\rangle$ denote a state of $AT$ in which $T$ is maximally entangled with the code space. The criterion for correctability says that if $R$ is a correctable region then for any operator $Y$ supported on $R$, $PYP = cP$; therefore, if $Y$ is supported on $R$ and $X$ is supported on $T$,
\begin{align}
%\langle \Phi| I_{R^c} \otimes Y_R \otimes X_T |\Phi\rangle= \langle \Phi| I_{R^c} \otimes Y_R \otimes X_T |\Phi\rangle\\
\langle \Phi|Y\otimes X |\Phi\rangle&= \langle \Phi|PYP\otimes X |\Phi\rangle
= c\langle \Phi| P\otimes X|\Phi\rangle\notag\\
& = c\langle \Phi|I\otimes X|\Phi\rangle = \langle \Phi| Y\otimes I|\Phi\rangle \langle \Phi|I\otimes X|\Phi\rangle.
\end{align}
Because $\langle \Phi|Y\otimes X|\Phi\rangle$ factorizes for any $Y$ supported on $R$ and $X$ supported on $T$, we conclude that the marginal density operator of $RT$ factorizes, 
\begin{equation}
\rho_{RT} = \rho_R\otimes \rho_T,
\end{equation}
if $R$ is correctable. 

To proceed we use properties of the entropy
\begin{equation}
S(\rho) = -{\rm tr}\left( \rho \log \rho \right),
\end{equation}
where for convenience we define entropy using logarithms with base $\dim \cH_0$.
Because  $R_1$ and $R_2$ are both correctable, $R_1T$ and $R_2T$ are product states; therefore
\begin{align}\label{eq:RT-products}
S(R_1T) = S(R_1) + S(T), \quad S(R_2T) = S(R_2) + S(T).
\end{align}
Denoting by $R^c$ the region of the code block complementary to $R_1R_2$, and noting that the overall state of $R_1R_2R^cT$ is pure, we have
\begin{align}\label{eq:purity-constraint}
S(R_1R^c) &= S(R_2T) = S(R_2) + S(T),\\
S(R_2R^c)& = S(R_1T) = S(R_1) + S(T);
\end{align}
adding these equations yields
\begin{align}
S(T) &= \frac{1}{2} \left( S(R_1R^c) + S(R_2R^c) - S(R_1)-S(R_2)\right)\\
& = S(R^c) -\frac{1}{2}\left(I(R_1;R^c) + I(R_2;R^c)\right).
\end{align}
Since the mutual information $I(R;R^c)$  is nonnegative (subadditivity of entropy), $S(T) = k$, and $S(R^c) \le |R^c| = n - |R_1| - |R_2|$, we obtain eq.(\ref{eq:disjoint-correctable}).
\end{proof}

\begin{corollary} [strong quantum Singleton bound] \label{coro:strong-singleton}
Consider a code subspace  $\cH_C= P\cH$, where $\cH$ contains $n$ sites, and where $k = \log \dim \cH_c / \log \dim \cH_0$. Then the distance $d$ and price $p$ of the code obey
\begin{equation}\label{eq:strong-singleton}
p - k \ge d - 1.
\end{equation}
\end{corollary}
\begin{proof}
In eq.(\ref{eq:disjoint-correctable}), choose $R_1$ to be the complement of the smallest region that supports the logical algebra of the code (hence $|R_1| = n - p$), and choose $R_2$ to be any set of $d-1$ qubits not contained in $R_1$. Then eq.(\ref{eq:strong-singleton}) follows.
\end{proof}

\begin{corollary} [quantum Singleton bound] \label{coro:singleton}
Consider a code subspace  $\cH_C= P\cH$, where $\cH$ contains $n$ sites, and where $k = \log \dim \cH_c / \log \dim \cH_0$. Then
\begin{equation}\label{eq:singleton}
n - k \geq 2 (d - 1)
\end{equation}
where $d$ is the code distance.
\end{corollary}
\begin{proof}
Combine Corollary \ref{coro:strong-singleton} and Lemma \ref{lemma:complementarity}.
\end{proof}

Because of its resemblance to the Singleton bound
\begin{equation}
n - k \geq d - 1,
\end{equation}
satisfied by classical $[n,k,d]$ codes, 
Eq.~(\ref{eq:singleton}) is called the {\it quantum Singleton bound}. We therefore call eq.(\ref{eq:strong-singleton}) the {\it strong quantum Singleton bound}. 
This bound is saturated by, for example, the [[7,1,3]] Steane code. In that case, the logical Pauli operators $\bar X$ and $\bar Z$ can both be supported on a set of three qubits; therefore the price is $p=3$ and the bound becomes
\begin{align}
1= k \le p - d + 1 = 3 - 3 +1 = 1. 
\end{align}

This strong quantum Singleton bound constrains the distance and price of a traditional subspace code, and it is natural to wonder what we can say about similar constraints on the distance and price of a logical subalgebra. 
In \S \ref{sec:holographic-singleton}, we will see that for holographic codes, using more sophisticated entropic arguments, we can derive an operator algebra version of the strong quantum Singleton bound.

\section{Holography and quantum error correction}\label{sec:Connection}

The AdS/CFT correspondence \cite{Maldacena1998} is a remarkable proposed equivalence between two theories --- quantum gravity in the bulk of a ($D$+1)-dimensional asymptotically anti-de Sitter spacetime, and conformally-invariant quantum field theory (CFT), without gravity, residing on the $D$-dimensional boundary of the spacetime. 
A very complex dictionary relates operators acting in the bulk theory to the corresponding operators in the boundary theory. 
This dictionary is only partially understood, but it is known that local operators acting deep inside the bulk correspond to highly nonlocal operators acting on the boundary. 
Much evidence indicates that geometrical properties of the bulk theory are intimately related to the structure of quantum entanglement in the boundary theory \cite{Ryu2006, VanRaamsdonk2010}. 
Further elucidation of this relationship should help to clarify how spacetime geometry can arise as an emergent property of a non-gravitational theory.

A puzzling feature of the correspondence is that a single bulk operator can be faithfully represented by a boundary operator in multiple ways. In a very insightful paper,  Almheiri, Dong, and Harlow \cite{Almheiri2015} suggested interpreting this ambiguity using the language of quantum error correction. According to their proposal, the low-energy sector of the boundary CFT can be viewed as a code subspace of the CFT Hilbert space, corresponding to weakly perturbed AdS geometry in the bulk, and the local operators acting on the bulk can be regarded as the logical operators acting on this code subspace. Local operators in the bulk can be reconstructed on the boundary in multiple ways, reflecting the property of quantum error-correcting codes that operators acting differently on the physical Hilbert space ${\cal H}$ may be logically equivalent when acting on the code subspace ${\cal H}_C$. High energy states of the CFT, which are outside the code space, correspond to large black holes in the bulk. 

In \cite{Pastawski2015}, {\it holographic codes} were constructed, which capture the features envisioned in \cite{Almheiri2015}. Such codes provide a highly idealized lattice regularization of the AdS/CFT correspondence, with bulk and boundary lattice sites. The code subspace, or bulk Hilbert space, is (disregarding some caveats expressed below) a tensor product of finite-dimensional Hilbert spaces, one associated with each bulk site, and likewise the boundary Hilbert space is a tensor product of finite-dimensional Hilbert spaces, one associated with each boundary site. The code defines an 
 embedding of the bulk Hilbert space inside the boundary Hilbert space. The embedding map can be realized by a tensor network construction based on a uniform tiling of the negatively curved bulk geometry. This tensor network provides an explicit holographic dictionary, in particular mapping each (logical) bulk local operator (with support on a single bulk site), to a corresponding physical nonlocal operator on the boundary (acting on many boundary sites). 

From the perspective of quantum coding theory, holographic codes are a family of quantum codes in which logical degrees of freedom have a pleasing geometrical interpretation, and as emphasized in \cite{Hayden2016}, this connection between coding and geometry can be extended beyond anti-de Sitter space. From the perspective of the AdS/CFT correspondence, holographic codes strengthen our intuition regarding how quantum error correction relates to emergent geometry. Both perspectives provide ample motivation for further developing these ideas. 

The precise sense in which the low-energy sector of a CFT realizes a quantum code remains rather murky. But loosely speaking the logical operators are CFT operators which map low-energy states to other low-energy states. Operators which are logically equivalent act on the low-energy states in the same way, but act differently on the high energy states which are outside the code space. The algebra of logical operators needs to be truncated, because acting on a state with a product of too many logical operators may raise the energy too high, so that the resulting state leaves the code space. 

From the bulk point of view, there is a logical algebra ${\cA}_x$ associated with each bulk site $x$, and formally the complete logical algebra of the code is, in first approximation,
\begin{align}\label{eq:bulk-product}
{\cal A} = \bigotimes_x {\cal A}_x,
\end{align}
where the tensor product is over all bulk sites. However, implicitly the number of bulk local operators needs to be small enough so that the back reaction on the geometry can be safely neglected. If so many bulk operators are applied that the bulk geometry is significantly perturbed, then the code space will need to enlarge, as explained \S \ref{subsec:beyond-finite}.  A further complication is that gauge symmetry in the bulk may prevent the bulk algebra from factorizing as in eq.(\ref{eq:bulk-product}).

The holographic dictionary determines how the logical operator subalgebra supported on a region in the bulk (a set of logical bulk sites) can be mapped to an operator algebra supported on a corresponding region on the boundary (a set of physical boundary sites). The geometrical interpretation of this relation between the bulk and boundary operator algebras will be elaborated in the following subsections. 

\subsection{Entanglement wedge reconstruction}\label{subsec:entanglement-wedge}

For holographic codes, whether a specified subsystem of the physical Hilbert space ${\cal H}$  is correctable with respect to a particular logical subalgebra can be formulated as a question about the bulk geometry. This connection between correctability and geometry is encapsulated by the {\it entanglement wedge hypothesis} \cite{Hubeny2007, Czech2012, Jafferis2016, Wall2014}, which holds in AdS/CFT \cite{Dong2016, Bao2016}. This hypothesis specifies the largest bulk region whose logical subalgebra can be represented on a given boundary region. 

The entanglement wedge hypothesis can be formulated for dynamical spacetimes, but for our purposes it will suffice to consider a special case. We consider a smooth Riemannian manifold $B$, which may be regarded as a spacelike slice through a static bulk spacetime. Somewhat more generally, we may imagine that $B$ is a slice though a Lorentzian manifold which is invariant under time reversal about $B$. Any $B$ can be locally extended to such a Lorentzian manifold which solves the Einstein field equation without matter sources. To formulate the entanglement wedge hypothesis for this case, we will need the concept of a minimal bulk surface embedded in $B$. We denote the boundary of $B$ by $\partial B$, and consider a boundary region $R\subseteq \partial B$. 
\begin{definition}[Minimal surface]\label{def:MinimalSurface}
Given a Riemannian manifold $B$ with boundary $\partial B$,  the {\it minimal surface} $\chi_R$ associated with a boundary region $R \subseteq \partial B$ is the minimum area co-dimension one surface in $B$ which separates $R$ from its boundary complement $R^c$ (see figure \ref{fig:MinimalSurfacesAndEWedge} for some examples).
\end{definition}
%The celebrated result by Ryu and Takayanagi \cite{Ryu2006} generalizes the Bekenstein-Hawking entropy of a black hole \cite{Bekenstein1973, Bardeen1973} and finds the entropy associated to the reduced density matrix $\rho_R$ of a boundary region $R$ to be proportional to the area of this minimal surface $\chi_R$.
We will, for the most part, assume that the minimal surface $\chi_R$ is unique and geometrically well defined.
Some choices of geometry $B$ and boundary region $R$ admit more than one minimal surface, but one may usually slightly alter the choice of $R$ in order to make the minimal surface $\chi_R$ unique.
Note that, according to Definition \ref{def:MinimalSurface}, $R$ and $R^c$ have the same minimal surface.
%, this is consistent with the recent interpretation of OAQEC of Harlow \cite{Harlow2016} where there is an area operator for $\chi_R$ which admits a representation in either of $\cH_R$ or $\cH_{R^c}$.

% Figure illustrating minimal surface and entanglement wedge
\begin{figure}[ht]
\begin{center}
\includegraphics[width=\columnwidth]{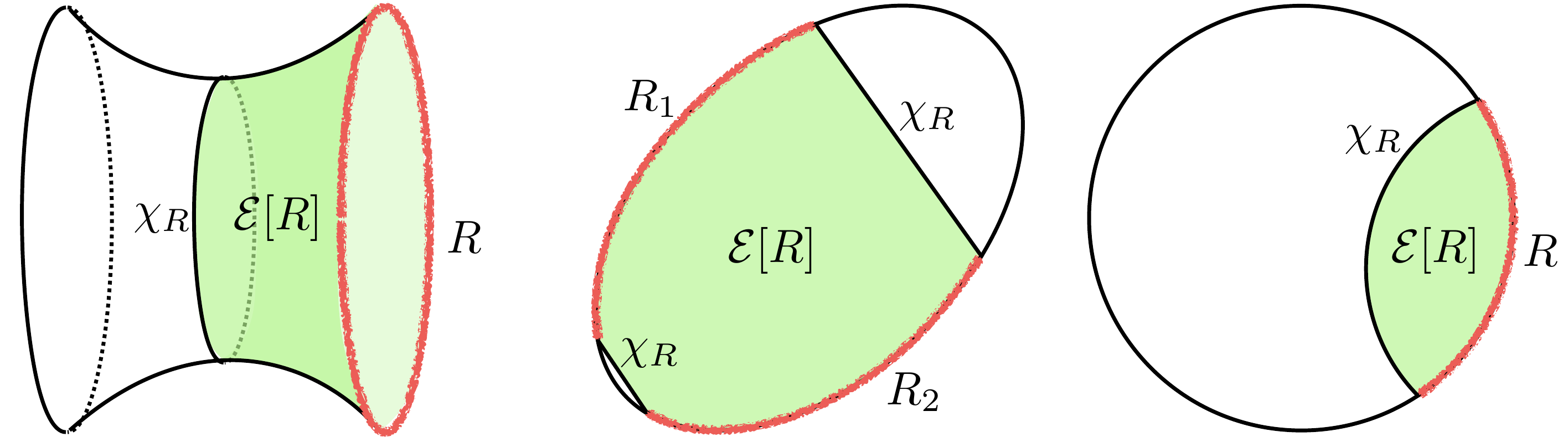}
\caption{
This figure illustrates  the geometric notions of minimal surface and entanglement wedge.
In each pane, we highlight a boundary region $R$ with a crayon stroke; 
the corresponding minimal surface $\chi_R$ is indicated, and the entanglement wedge $\cE[R]$ is shaded in green.
On the left $B$ is a hyperboloid whose boundary $\partial B$ has two connected components, where $R$ is one of those components (the one on the right). The minimal surface cuts the hyperboloid at its waist, and the entanglement wedge is everything to the right of $\chi_R$.
In the central pane $B$ is the interior of a Euclidean ellipse; the boundary region $R= R_1 \sqcup R_2$ has two connected components, and $\chi_R$ also has two connected components. As shown, the connected components of $\chi_R$ need not be homologous to $R_1$ and $R_2$, allowing $\cE[R_1 \sqcup R_2]$ to be significantly larger than $\cE[R_1] \sqcup \cE[R_2]$.
On the right $B$ is the Poincar\'e disc, portraying an infinite hyperbolic geometry.
%Here, although  $\partial B$ is found at infinity, it is graphically presented as a constant radius circle. 
The minimal surface is a geodesic in the bulk with endpoints on $\partial B$.
}\label{fig:MinimalSurfacesAndEWedge}
\end{center}
\end{figure}

%Entanglement wedge definition
Now we can define the entanglement wedge.
\begin{definition}[Entanglement wedge]\label{def:EntanglementWedge} Given a boundary region $R \subseteq \partial B$, the {\it entanglement wedge} of $R$ is a bulk region $\mathcal{E}[R] \subseteq B$, whose boundary is  $\partial  \cE[ R] \assign \chi_R \cup R$, where $\chi_R$ is the minimal surface for $R$ (see figure \ref{fig:MinimalSurfacesAndEWedge} for some examples).
\end{definition}
Note that under the uniqueness assumption for the minimal surface $\chi_R$, the entanglement wedge $\cE(R)$ of a boundary region $R$ and the entanglement wedge $\cE(R^c)$  of its boundary complement $R^c$ cover the full bulk manifold $B$, and they intersect exclusively at the minimal surface.
\begin{hypothesis}[Geometric complementarity]\label{hyp:GeometricComplementarity}
  Given a region $R \subseteq \partial B$ and its boundary complement $R^c$ we have that $\chi_R = \chi_{R^c} =\mathcal{E} [ R] \cap
  \mathcal{E} [ R^c]$ and $\mathcal{E} [ R] \cup \mathcal{E} [ R^c] = B$.
\end{hypothesis}
As we will see, this geometric statement, which holds for a generic manifold $B$ and boundary region $R$, leads to very strong code-theoretic guarantees under the entanglement wedge hypothesis.
%The fact that the minimal surface changes discontinuously implies that this hypothesis can not hold for arbitrary $R$. However, we shall allow ourselves to disregard such choices of $R$, which are non-generic in the continuum.

% Entanglement wedge hypothesis
For a holographic code, the entanglement wedge hypothesis states a sufficient condition for a boundary region to be correctable with respect to the logical subalgebra supported at a site in the bulk. Due to Lemma \ref{lemma:reconstruction}, this condition also informs us that the logical subalgebra can be reconstructed on the complementary boundary region. Evoking the continuum limit of the regulated bulk theory, we will sometimes refer to a bulk site as a {\it point} in the bulk, though it will be implicit that associated logical subalgebra is finite dimensional and slightly smeared in space. 
\begin{hypothesis}[Entanglement wedge hypothesis]\label{hyp:Entanglement Wedge}
  If the bulk point $x$ is contained in the entanglement wedge $\mathcal{E} [ R]$ of boundary region $R$, then the complementary boundary region $R^c$ is correctable with respect to the logical bulk subalgebra  $\cA_x$. Thus for each operator in $\cA_x$, there is a logically equivalent operator supported on $R$. 
 \end{hypothesis}
 This connection between holographic duality and operator algebra quantum error correction has many implications worth exploring.
 
 For a holographic code corresponding to a regulated boundary theory, there are a finite number of boundary sites, each describing a finite-dimensional subsystem. Thus we can speak of the length $n$ of the code, meaning the number of boundary sites, as well as the distance $d$ and price $p$ of the code (or of any logical subalgebra), which also take integer values. It is convenient, though, to imagine taking a formal continuum limit of the boundary theory in which the total boundary volume stays fixed as  $n\to \infty$, while maintaining a uniform number of boundary sites per unit boundary volume as determined by the bulk induced metric. Without intending to place restrictions on the dimension of $B$, from now on we will use the term {\it area} when speaking about the size of a boundary region, and save the term {\it volume} for describing the size of a bulk region. In the continuum limit, we may still speak of $n$, $d$, and $p$, but now taking real values;  $n$ becomes the total area of the boundary, while $d(\cA)$ is the area of the smallest boundary region which is not correctable with respect to logical subalgebra $\cA$, and $p(\cA)$ is the area of the smallest boundary region which supports $\cA$. For now, to ensure that back reaction on the bulk geometry is negligible, we will suppose that the bulk algebra $\cA$ has support on a constant number of points. In the formal continuum limit, then, the logical algebra has negligible dimension, in effect defining a $k\approx 0$ code if the size of the logical system is expressed in geometrical units. 

The entanglement wedge hypothesis has notable consequences for the logical subalgebra $\cA_x$  supported at a bulk point $x$. Consider the distance of $\cA_x$.  For any boundary region $R$ with boundary complement $R^c$, if $x\in \cE[R^c]$ then $R$ is correctable with respect to $\cA_x$. On the other hand, if $x\in \cE[R]$, then $\cA_x$ can be reconstructed on $R$; arguing as in the proof of Lemma \ref{Lemma:NoFreeLunch}, $R$ cannot be correctable with respect to $\cA_x$ if $R$ supports $\cA_x$ and $\cA_x$ is non-abelian. By geometric complementarity, either $x\in \cE[R^c]$ or $x \in \cE[R]$; we conclude that $R$ is correctable with respect to $\cA_x$ if and only if $x \in \cE[R^c]$. By the definition of distance, then,
\begin{align}\label{eq:point-distance}
 d(\cA_x) = \min_{R\subseteq \partial B: x\not \in \mathcal{E}(R^c)} |R|,
\end{align}
assuming $\cA_x$ is non-abelian. 

Now consider the price of $\cA_x$. According to the entanglement wedge hypothesis, $\cA_x$ can be reconstructed on $R$ if $x\in\cE[R]$. On the other hand, if $x\not \in \cE[R]$, then $x\in\cE[R^c]$ by geometric complementarity, and therefore $\cA_x$ can be reconstructed on $R^c$. Because operators supported on $R$ commute with operators supported on $R^c$, it is not possible for $\cA_x$ to be reconstructed on both $R$ and $R^c$ if $\cA_x$ is non-abelian. We conclude that $\cA_x$ can be reconstructed on $R$ if and only if $x\in\cE[R]$. By the definition of price, then, 
\begin{align}\label{eq:point-price}
p(\cA_x) = \min_{R\subseteq \partial B: x\in \mathcal{E}(R)} |R|,
\end{align}
assuming $\cA_x$ is non-abelian. 

Geometric complementarity says that $x\in \cE[R]$ if and only if $x\not\in \cE[R^c]$. Therefore by comparing eq.(\ref{eq:point-distance}) and eq.(\ref{eq:point-price}), we see that the expressions for the distance and the price are identical. Thus we have shown:
%%%%%
\begin{lemma}[price equals distance for a point]\label{lemma:Complementarity} For a holographic code, let $\cA_x$ be the non-abelian logical algebra associated with a bulk point $x$.
Then 
\begin{align}
p(\cA_x) = d(\cA_x) .
\end{align}
\end{lemma}
%%%%%
Thus, in a holographic code, the bound $p(\cA_x) \ge d(\cA_x)$ in Lemma \ref{Lemma:NoFreeLunch} is saturated by the logical subalgebra of a point. It is intriguing that a geometrical point admits this simple algebraic characterization, suggesting how geometrical properties might be ascribed to logical subalgebras in a broader setting. 

We can extend this reasoning to a bulk region $X$ which contains a finite number of bulk points, continuing to assume that the number of operator insertions is sufficiently small that back reaction on the bulk geometry can be neglected, and that the logical subalgebra factorizes as in eq.(\ref{eq:bulk-product}). In that case, a boundary region $R$ is correctable with respect to the algebra $\cA_X$ if it is correctable with respect to $\cA_x$ for each bulk point $x$ in $X$. Therefore,
\begin{align}\label{eq:distanceMultipoint}
d(\cA_{X}) = \min_{x \in X} d(\cA_{x}).
\end{align}
Eq.(\ref{eq:point-price}) can also be extended to a bulk region $X$:
\begin{align}\label{eq:priceMultipoint}
p(\cA_{X}) \assign \min_{R \subseteq \partial B: X\subseteq \cE[R] } |R| ,
\end{align}
assuming that each nontrivial operator in $\cA_X$ fails to commute with some other operator in $\cA_X$. If all points of $X$ are contained in $\cE[R]$, it follows that $p(\cA_X)\le |R|$.

We emphasize again that these properties apply not only to AdS bulk geometry, but also to other quantum code constructions satisfying geometric complementarity and the entanglement wedge hypothesis. Such codes were constructed in  \cite{Pastawski2015} for tensor networks associated with tilings of bulk geometries having non-positive curvature. These results were extended to arbitrary graph connectivity in \cite{Hayden2016}, where a discrete generalization of the entanglement wedge hypothesis was found to be valid in the limit of large bond dimension. Taking a suitable limit, these codes can be viewed as regularized approximations to underlying smooth geometries. 

\subsection{Punctures in the bulk}\label{subsec:beyond-finite}

In quantum gravity, there is an upper limit on the dimension of the Hilbert space that can be encoded in a physical region known as the Bousso bound \cite{Bousso1999}; the log of this maximal dimension is proportional to the surface area of the region.
When one attempts to surpass this limit, a black hole forms, with entropy proportional to the area of its event horizon.

This feature of bulk quantum gravity can be captured by holographic codes, rather crudely, if we allow punctures in the bulk. A subsystem of the code space $\cH_C$  resides along the edge of each such  puncture, and the holographic tensor network provides an isometric embedding of this logical subsystem in the physical Hilbert space $\cH$ which resides on the exterior boundary of the bulk geometry. This picture is crude because in an actual gravitational theory a black hole in the bulk would carry mass and modify the bulk curvature outside the black hole. For our purposes this impact on the curvature associated with a puncture will not be particularly relevant and we will for the most part ignore it here.

In the continuum limit, we associate the holographic code with a Riemannian bulk manifold $B$ as in \S \ref{subsec:entanglement-wedge}, but now the boundary $\partial B$ is the union of two components:  the exterior (physical) boundary, denoted $\Phi$, and the interior (logical) boundary, denoted $\Lambda$. The logical boundary is the union of the boundaries of all punctures. For the physical region $R\subseteq \Phi$, we will now need to distinguish between its {\it boundary complement} $ \partial B \setminus R$ and its {\it physical complement} $\Phi \setminus R$. The entanglement wedge hypothesis continues to apply, where now it is understood that the minimal surface $\chi_R$ separates $R$ from its boundary complement. In AdS/CFT this is called the {\it homology constraint}, meaning that $R\sqcup \chi_R$ is the boundary of a bulk region. 

As we take the continuum limit $n\to \infty$ with the bulk geometry fixed, we assume as before that the density of sites per unit area on $\partial B$ is uniform, with the same density on both the physical boundary and the logical boundary. If we assume as in \S \ref{subsec:entanglement-wedge} that the bulk logical algebra outside of the punctures is supported on a bounded number of bulk points, then the size $k$ of the logical system  is determined by the area of the logical boundary:
\begin{align}\label{eq:k-over-n}
k / n = |\Lambda| / |\Phi|.
\end{align}

A holographic code without punctures obeys the celebrated Ryu-Takayanagi formula \cite{Ryu2006}, which asserts that for a boundary region $R$, the entanglement entropy of $R$ with its physical  complement $R^c := \Phi\setminus R$ is the area $|\chi_R|$ of the minimal surface separating $R$ from $R^c$, with area measured in the same units used to define $|\Phi|$ and $|\Lambda|$. To extend this formula to a manifold $B$ with punctures, we imagine introducing a reference system $T$ which is maximally entangled with the logical system $\Lambda$, so that the joint state of $T \sqcup \Phi$ is pure. 
For a physical boundary region $R\subseteq \Phi$, the area $|\chi_R|$  of the minimal surface separating $R$ from its boundary complement $\Lambda \sqcup (\Phi \setminus R)$ is the entanglement entropy of $R$ with $T \sqcup (\Phi \setminus R)$ in this pure state.

One way to visualize the purifying reference system $T$ is inspired by the thermofield double construction used in AdS/CFT \cite{Maldacena2003}. 
Given a manifold $B$ with physical boundary $\Phi$ and logical boundary $\Lambda$, we introduce a second copy $\tilde B$ of the manifold, with physical boundary $\tilde \Phi$ and logical boundary $\tilde \Lambda$ (see figure \ref{fig:BulkDoubling}). Then we join $\Lambda$ and $\tilde \Lambda$, obtaining manifold $B\tilde B$, whose physical boundary $\Phi\tilde\Phi$ has two connected components. This construction describes two holographic codes, whose logical systems are maximally entangled; the second copy of the code provides the reference system purifying the first copy. 
The combined manifold $B\tilde B$ has no punctures, and we can apply the original formulation of the Ryu-Takayanagi formula to $B\tilde B$. 
For a boundary region $R\subseteq \Phi$, the minimal surface $\chi_R$ separating $R$ from $\tilde \Phi \sqcup (\Phi\setminus R)$ never reaches into $\tilde B$, and therefore coincides  with the minimal surface separating $R$ from $\Lambda \sqcup (\Phi \setminus R)$.
%%%%
(We note that if the logical boundary $\Lambda$ represents the event horizon of a static (2+1)-dimensional black hole, then, because a geodesic outside the black hole is the spatial trajectory of a light ray, $\chi_R$ either fully contains $\Lambda$ or avoids it entirely. For a nonstatic geometry, it is possible for $\chi_R$ to include only part of $\Lambda$.)
%%%%
\begin{figure}[ht]
	\begin{center}
		\includegraphics[width=\columnwidth]{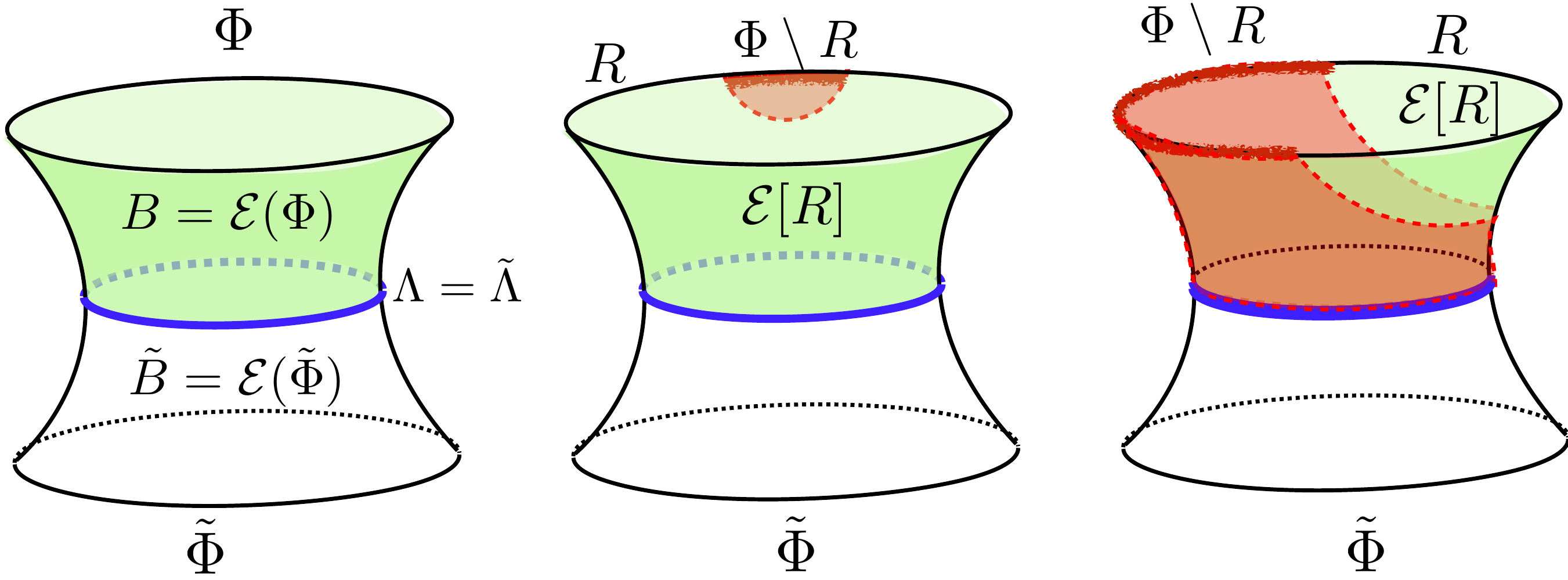}
		\caption{The left pane illustrates the thermofield double construction, in which a bulk manifold $B$ with logical boundary $\Lambda$ is extended to a two copies of $B$ with their logical boundaries identified. This doubled manifold $B\tilde B$ describes two holographic codes whose logical systems are maximally entangled. The other two panes illustrate that, for a boundary region $R$ contained in the physical boundary $\Phi$ of $B$, the corresponding minimal surface lies in $B$. 
}\label{fig:BulkDoubling}
	\end{center}
\end{figure}

The Ryu-Takayanagi formula relating entanglement entropy $S(R)$ to $|\chi_R|$ is actually the leading term in a systematic expansion \cite{Faulkner2013}, in which the next correction arises from entanglement among bulk degrees of freedom, specifically the bulk entanglement of $\cE(R)$ with its bulk complement. 
In fact the division of $S(R)$ into the geometrical contribution $|\chi_R|$ and the bulk entanglement contribution is not a renormalization group invariant; the geometrical contribution dominates in the extreme low-energy limit of the boundary theory, and the bulk entropy becomes more important as we probe the boundary theory at shorter and shorter distances. We implicitly work in the low-energy limit in which geometrical entanglement is dominant. Even in this framework, it is possible to include bulk entanglement in our discussion, in accord with the so-called ER=EPR principle \cite{Maldacena2013}, which identifies entanglement with wormhole connectedness. We may, for example, consider two punctures of equal size in the bulk, and identify their logical boundaries $\Lambda_1$ and $\Lambda_2$, so that the two logical systems are maximally entangled. If $R$ is sufficiently large, the minimal surface $\chi_R$ in the bulk might pass in between the two punctures, and include (say) the logical boundary $\Lambda_1$ in order to satisfy the homology constraint. Then the entanglement entropy of $R$ includes a contribution $|\Lambda_1|$ due to inclusion of the puncture in its entanglement wedge $\cE[R]$. In this way, the geometrical entanglement of $R$ can capture the bulk entanglement shared by the punctures. 

Up until now we have implicitly assumed that the holographic code provides an isometric embedding of the logical system $\Lambda$ into the physical system $\Phi$. This requirement places restrictions on the geometry of the puncture(s). The Ryu-Takayanagi formula provides one possible way to understand the  restriction. Suppose that the reference system $T$ is maximally entangled with the logical boundary $\Lambda$, so that the state of $T\sqcup \Phi$ is pure, and the entanglement entropy of $\Phi$ is $|\Lambda|$. This must agree with the geometrical entropy given by $|\chi_\Phi| = |\chi_\Lambda|$, the area of the minimal surface separating $\Phi$ and $\Lambda$. Holographic tensor network constructions suggest a stronger constraint,
\begin{align}\label{eq:chi=equals=Lambda}
\chi_\Phi = \chi_\Lambda = \Lambda,
\end{align}
since in that case the tensor network provides an explicit isometric map from $\Lambda$ into $\Phi$. This consistency condition is illustrated in figure \ref{fig:LogicalPhysical}. The constraint eq.(\ref{eq:chi=equals=Lambda})  ensures that the geometrical entanglement entropy $S(\Lambda)$ is compatible with the Bekenstein-Hawking entropy $|\Lambda|$ of a black hole with event horizon at $\Lambda$, as we should expect when the microstates of the black hole are maximally entangled with a reference system. 
If several such black holes approach one another, they must coalesce into a larger black hole in order to enforce  eq.(\ref{eq:chi=equals=Lambda}). See figure \ref{fig:LogicalPhysical}.
\begin{figure}[ht]
\begin{center}
\includegraphics[width=\columnwidth]{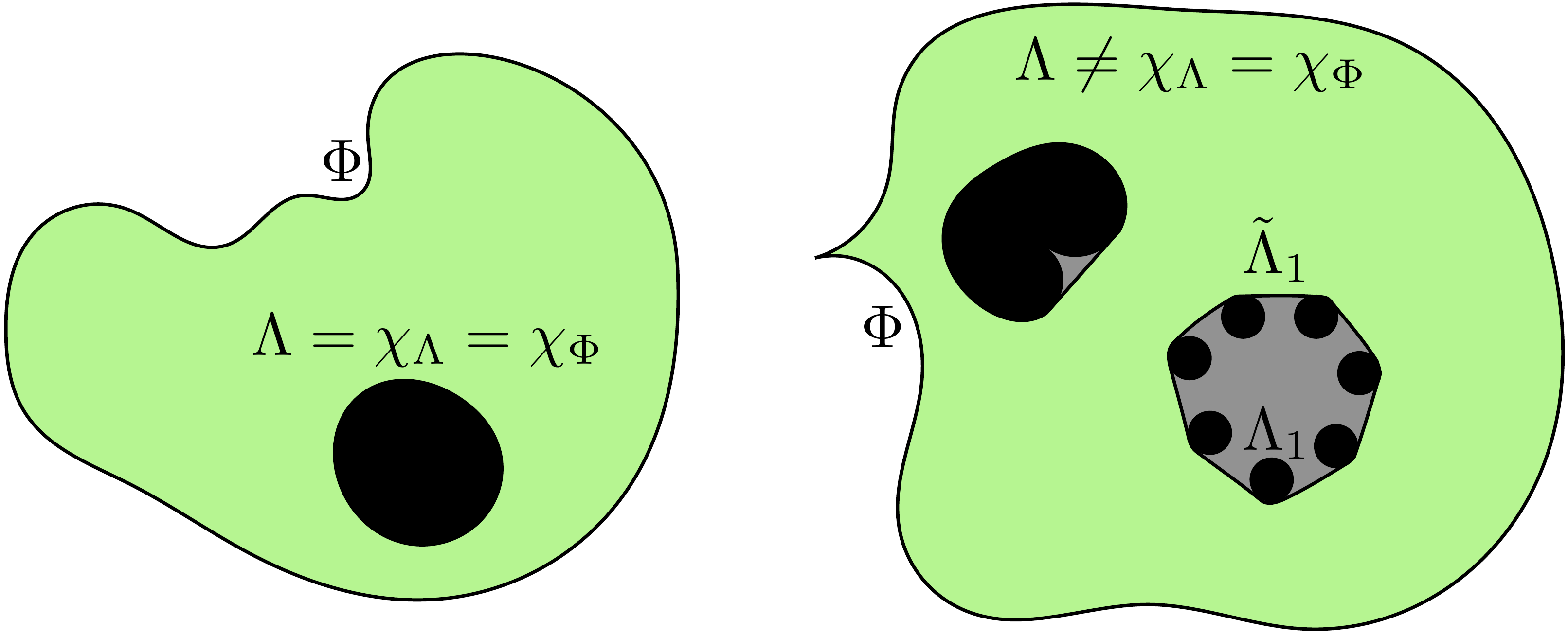}
\caption{
This figure illustrates the necessary condition $\chi_\Lambda= \Lambda$ for the interior boundary of a Riemannian manifold $B$ to be identified as a logical system.  
In both panes, the physical Hilbert space $\cH$ resides on the exterior boundary  $\Phi$ of $B$, and $\Lambda$ is the boundary of the punctures in the bulk, which are shaded in black.  
The green region is the entanglement wedge $\cE[\Phi] $, bounded by $\Phi$ and the minimal surface $\chi_\Phi= \chi_\Lambda$ separating $\Phi$ from $\Lambda$; the gray region is $B \setminus \cE[\Phi]$. 
For purposes of illustration we assume the bulk metric is Euclidean. 
On the left, we have $\chi_\Lambda = \Lambda$ and the interpretation of $\Lambda$ as a logical system is consistent. On the right, we have $\chi_\Lambda \neq \Lambda$ and two possible reasons for this are illustrated. First, a connected component of  $\Lambda$ may fail to be convex. 
Second, the union $\Lambda_1$ of several connected components of $\Lambda$  may be encapsulated by a surface $\tilde \Lambda_1$ with smaller area than $\Lambda_1$, in which case the logical system resides on $\tilde \Lambda_1$ rather than $\Lambda_1$. 
The emergence of this new logical system is reminiscent of the merging of small black holes to form a larger black hole. 
}\label{fig:LogicalPhysical}
\end{center}
\end{figure}

For a holographic code with punctures, 
%or some other form of logical boundary $\Lambda$, 
we may consider the logical subalgebra associated with a bulk region $X\subseteq B$, where now $X$ may include pieces of $\Lambda$. 
Because we are considering the low-energy regime in which geometric entanglement dominates bulk entanglement and back reaction on the geometry due to bulk fields is negligible, we ignore the contribution of bulk points to the logical subalgebra, just as in eq.(\ref{eq:k-over-n}).
Therefore, if $X$ does intersect with $\Lambda$, then the effective size of the logical system is given by
\begin{align}
k_X = |\Lambda \cap X|,
\end{align}
the area of the portion of $\Lambda$ contained in $X$.
(In the language of holographic tensor network codes \cite{Pastawski2015}, we are assuming that most bulk tensors carry no logical indices, so that nearly all of the bulk logical indices contained in the bulk region $X$ are located on the logical boundary.) 

For bulk region $X$, we may consider a reference system $T$ which is maximally entangled with the logical subsystem residing in $X\subseteq B$. 
Then when we say that $\cA_X$ can be reconstructed on boundary region $R\subseteq \Phi$, we mean that $R$ contains a subsystem which is maximally entangled with $T$. 
If we apply the entanglement wedge hypothesis to a manifold $B$ with a logical boundary, we can use the same reasoning as in \S \ref{subsec:entanglement-wedge} to obtain a geometrical expression for the price of the logical algebra:
\begin{align}
p(\cA_X) \assign \min_{R \subseteq \Phi: X\subseteq \mathcal{E}[R] } |R|,
\end{align}
On the other hand, a boundary region $R\subseteq{\Phi}$ will be correctable with respect to $\cA_X$ if $\cA_X$ is supported on the physical complement $\Phi\setminus R$ of $R$, and the expression for distance becomes
\begin{align}
d(\cA_X) \assign \min_{R \subseteq \Phi: X\not \subseteq\mathcal{E}[\Phi \setminus R]} |R|. 
\end{align}
It is important to notice that the minimization in $d(\cA_X)$ is over $X$ not contained in the entanglement wedge of the physical complement, rather than the boundary complement, of $R$. In particular, when $X$ is a single point $\{x\}$ in the bulk, the expressions for price and distance are not identical because the entanglement wedges $\cE[R]$ and $\cE[\Phi\setminus R]$ are not complementary regions of the bulk. Therefore Lemma \ref{lemma:Complementarity} does not apply to the case of a bulk manifold $B$ with logical boundaries. 

The geometrical interpretations for price and distance of $\cA_X$ allow us to prove a version of the strong quantum Singleton bound that applies to subalgebras with nonvanishing $k_X$ in the continuum limit. This will be explained in \S \ref{sec:holographic-singleton}.

\section{Negative curvature and uberholography}\label{sec:uberholography}

Next we discuss a general property of holographic codes defined on bulk manifolds with asymptotically uniform negative curvature, which we call {\it uberholography}. The essence of uberholography is that both the distance and price of a logical subalgebra scale sublinearly with the length $n$ of the holographic code. In the formal continuum limit $n\to \infty$, the logical subalgebra can be supported on a fractal subset of the boundary, with fractal dimension strictly less than the dimension of the boundary. This fractal dimension is a universal feature of the code, in the sense that it does not depend on which logical subalgebra we consider. Uberholography is intriguing, as it suggests that $(D+1)$-dimensional bulk geometry can emerge, not just from an underlying $D$-dimensional system, but also from a system of even lower dimension. 

Though uberholography applies more generally, to be concrete we consider the bulk to have a two-dimensional hyperbolic geometry with radius of curvature $L$.  Now the boundary is one-dimensional, and the minimal ``surface'' $\chi_R$ associated with connected boundary region $R$ is really a bulk geodesic, whose ``area'' is actually the geodesic's length.  For our purpose we need to know only one feature of the bulk geometry: For an interval $R$ on the boundary with length $|R|$, the length of the bulk geodesic $\chi_R$ separating $R$ from its boundary complement is
\begin{align}\label{eq:chi-hyperbolic}
|\chi_R| = 2L \log(|R|/a)
\end{align}
Here $a$ is a short-distance cutoff, which we may think of as a lattice spacing for the boundary theory, so that $|R|/a$ is the number of boundary sites contained in $R$. Applying the Ryu-Takayanagi formula, we conclude that the entanglement entropy $S(R)$ scales logarithmically with the size of $R$, which is the expected result for the vacuum state of a CFT in one spatial dimension. 

For some bulk region $X$, we would like to compute the distance of the logical subalgebra $\cA_X$ associated with $X$. This distance $d(\cA_X)$ is the size of the smallest boundary region $R$ which is not correctable with respect to $X$. Pick a point $x$ in $X$, and choose a connected boundary region $R$ such that $\cE[R]$ contains $x$, but just barely --- if we choose a slightly smaller connected boundary region $R'\subset R$, then $\cE[R']$ will not contain $x$. Since $x\in \cE(R)$, we know that $R$ is not correctable with respect to $\cA_x$, and therefore $d(\cA_X) \le |R|$. We could get a tighter upper bound on $d(\cA_X)$ if we could find a smaller boundary region $R'\subset R$ whose entanglement wedge still contains $x$. There may be no such connected boundary region, but can we  find a disconnected $R'\subset R$ such that $x\in \cE[R']$? 

\begin{figure}[t]
\begin{center}
\includegraphics[width=\columnwidth]{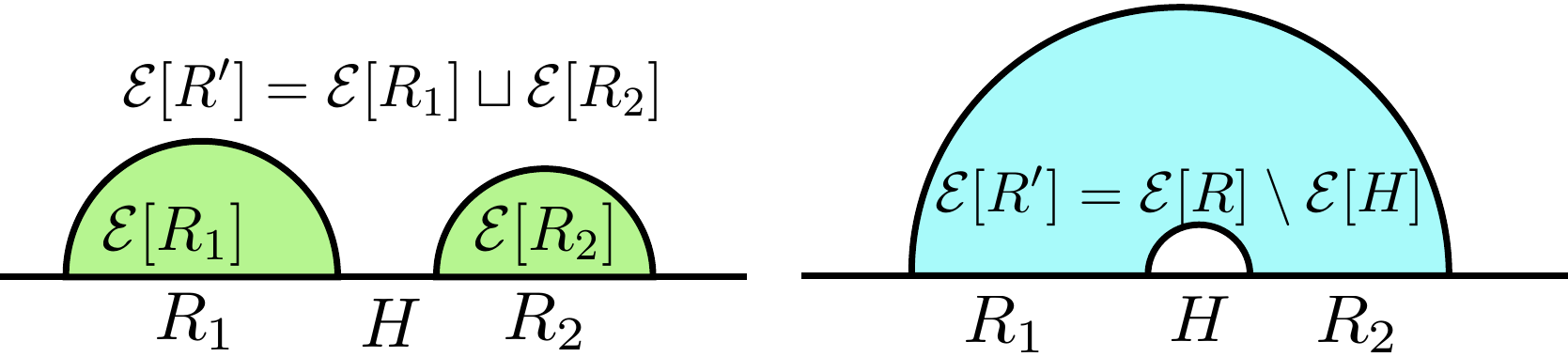}
\caption{
This figure illustrates the two possible geometries for the entanglement wedge $\cE[R']$ of a boundary region $R' = R_1\sqcup R_2$ with two connected components separated by the interval $H$. In the left pane, the minimal surface is $\chi_{R'} = \chi_{R_1} \sqcup \chi_{R_2}$ and the entanglement wedge is $\cE[R'] = \cE[R_1] \sqcup \cE[R_2]$. In the right pane, the minimal surface is $\chi_{R'} = \chi_R \sqcup \chi_H$, where $R= R_1HR_2$, and the entanglement wedge is  $\cE[R'] = \cE[R] \setminus \cE[H]$. 
}\label{fig:EntanglementWedge}
\end{center}
\end{figure}

Let's try punching a hole in $R$. That is, we divide $R$ into three consecutive disjoint intervals $R_1 H R_2$, where 
\begin{align}
|R_1| = |R_2| = \left(\frac{r}{2}\right)|R|, \quad |H| = (1-r)|R|, \quad 0 < r < 1,
\end{align}
and then remove the middle (hole) interval $H$, leaving the disconnected region $R' = R_1 R_2 = R\setminus H$. There are two possible ways to choose bulk geodesics which separate $R'$ from its complement (illustrated in figure \ref{fig:EntanglementWedge}), either $\chi_{R_1} \sqcup \chi_{R_2}$ or $\chi_R \sqcup \chi_H$; the minimal surface $\chi_{R'}$ is the smaller of these two. Thus if 
\begin{align}\label{eq:H-remove-inequality}
|\chi_{R_1}|+|\chi_{R_2}| > |\chi_R| + |\chi_H|,
\end{align}
we'll have
\begin{align}
{\cal E}[R'] = {\cal E}[R] \setminus {\cal E}[H];
\end{align}
removing $H$ from $R$ has the effect of removing ${\cal E}[H]$ from $\cE[R]$. Therefore $\cE[R']$ still contains $x$, and hence $d(\cA_X) \le |R'|$.

If we choose $H$ as large as possible, while respecting eq.(\ref{eq:H-remove-inequality}), then eq.(\ref{eq:chi-hyperbolic}) implies 
\begin{align}
|R_1|\cdot |R_2| = |R|\cdot |H| \implies r^2/4 = (1-r),
\end{align}
which is satisfied by 
\begin{align}
r/2 = \sqrt{2} - 1.
\end{align}
Each connected component of $R'$ is smaller than $R$ by this factor.

Now we repeat this construction recursively. In each round of the procedure we start with a disconnected region $\tilde R$ such that $\cE[\tilde R]$ contains $x$, where $\tilde R$ is the union of many connected components of equal size. Then we punch a hole out of each connected component to obtain a new region $\tilde R'$ such that $\cE[\tilde R']$ still contains $x$. Punching the holes increases the number of connected components by a factor of 2, and reduces the size of each component by the factor $r/2$. 

The procedure halts when the connected components are reduced in size to the lattice spacing $a$, which occurs after $m$ rounds, where 
\begin{align}
a = (r/2)^m |R|.
\end{align}
The remaining region $R_{\rm min}$ has $2^m$ components, each containing one lattice site. so that 
\begin{equation}\label{eq:fractal-bound}
d(\cA_X) \le |R_{\rm min}|/a = 2^m = (|R|/a)^\alpha,
\end{equation}
where
\begin{align}
\alpha =  \frac{\log 2}{\log (2/r)} = \frac{1}{\log_2(\sqrt{2}+1)} \approx .786.
\end{align}
The initial interval $R$ is surely no larger than the whole boundary, so the distance is bounded above by $n^\alpha$ for any logical subalgebra, where $d$ and $n$  are expressed as a number of boundary sites (rather than length along the boundary). 

We can also consider codes with punctures in the bulk. To be specific, suppose $B$ is a hyperbolic disk of proper radius $r_{\rm out}$, with a single puncture at the center of radius $r_{\rm in}$. The code length $n$ is proportional to the circumference of the outer boundary, and the size $k$ of the logical system is proportional to the circumference of the inner boundary. Because the circumference of a circle with radius $r$ is $2\pi L e^{r/L}$, the rate of the code is 
\begin{align}
k/n = e^{(r_{\rm in} - r_{\rm out})/L}.
\end{align}
We may choose an interval $R$ on the boundary, such that $\chi_R$ is tangent to the inner boundary at a single point. The length of this geodesic is essentially twice the difference between the inner and outer boundaries, so that eq.(\ref{eq:chi-hyperbolic}) implies
\begin{align}
r_{\rm out} - r_{\rm in} = L \log (|R|/a)
\end{align}
Using the recursive construction to repeatedly carve holes out of $R$, we obtain the bound eq.(\ref{eq:fractal-bound}) on the code distance, which becomes
\begin{align}
d \le (|R|/a)^\alpha = \left(e^{(r_{\rm out} - r_{\rm in})/L}\right)^\alpha = (n/k)^\alpha
\end{align}
(with the code distance expressed as a number of boundary sites). This scaling of the code distance, with $\alpha \approx .786$, compares favorably with the bound \cite{Bravyi2010b} on local commuting projector codes defined on a two-dimensional Euclidean lattice, for which $\alpha = 1/2$.

The scaling $p(\cA_X) \sim n^\alpha$ applies to price as well as distance. Once we have found a sufficiently large boundary region $R$ such that $\cE[R]$ contains the bulk region $X$, we can proceed to hollow out $R$ recursively until we reach the much smaller region $R_{\rm min}$ such that $|R_{\rm min}|/a = (|R|/a)^\alpha$ where $\cE[R_{\rm min}]$ still contains $X$, and hence $\cA_X$ is supported on $R_{\rm min}$. The resulting region $R_{\rm min}$, with fractal dimension $\alpha$, has a geometry reminiscent of the Cantor set, as illustrated in figure \ref{fig:CantorWedge}. 

\begin{figure}[ht]
\begin{center}
\includegraphics[width=0.8\columnwidth]{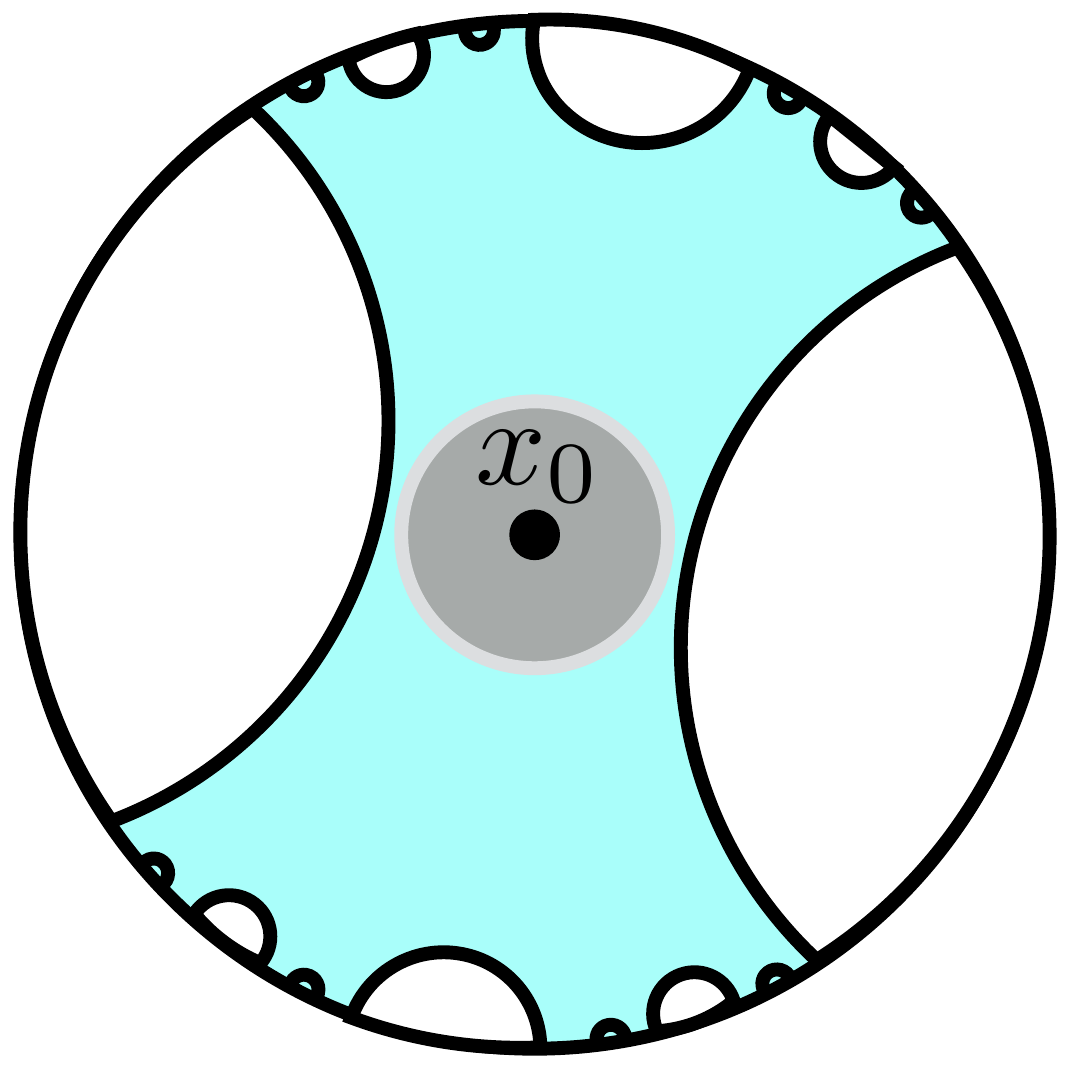}
\caption{
This figure illustrates uberholography for the case of a two-dimensional hyperbolic bulk geometry. The inner logical boundary is contained inside the entanglement wedge, shaded in blue, of a boundary region $R$. By repeatedly punching holes of decreasing size out of this boundary region, we obtain a much smaller region $R_{\rm min}$ whose entanglement wedge still contains the logical boundary. Thus the logical algebra is supported on a fractal boundary set, whose geometry is reminiscent of the Cantor set. 
}\label{fig:CantorWedge}
\end{center}
\end{figure}

It is interesting to compare this universal exponent $\alpha$ for planar uberholography with  the scaling laws for distance and price realized by other code families, such as holographic tensor network codes and concatenated quantum codes. In fact it can be quite challenging to obtain tight lower bounds on distance and price for tensor network code constructions; see Appendix \ref{app:Comparison} for further discussion.

\section{Quantum Markov condition and local correctability}\label{sec:quantum-Markov}

For a holographic code, consider (as in \S \ref{sec:uberholography}) a connected region $R= R_1 H R_2$ which is the disjoint union of three adjoining intervals. Imagine that the middle interval $H$ is erased. If $H$ is correctable, there is a recovery map $\cR$ which corrects this erasure error. But now we ask whether a stronger condition is satisfied: Is it possible to choose a recovery map taking $R'=R_1R_2$ to $R$, so that $\cR^{R'\to R'H}$ ``fills in'' the erased hole $H$? If the erasure of $H$ can be corrected by a map which acts only on a somewhat larger region containing $H$ (larger by a constant factor independent of system size), then we say that erasure is {\it locally correctable}.

The {\it quantum Markov condition} provides a criterion for local correctability \cite{Flammia2016}.  
We say that the state $\rho_{ABC}$ of three disjoint regions $A$, $B$, $C$ obeys the  quantum Markov condition (also called {\it quantum conditional independence}) if
\begin{align}\label{eq:CMI-zero}
0= I(A;C|B) = S(AB) + S(BC) - S(ABC) - S(B),
\end{align}
which is equivalent to saying that the strong subadditivity inequality is saturated (satisfied as an equality). 
If the Markov condition is satisfied, then $\rho_{ABC}$ can be reconstructed from the marginal state $\rho_{AB}$ using a map ${\cal R}^{B\to BC}$ which maps $B\to BC$:
\begin{align}
{\cal R}^{B\to BC}: \rho_{AB} \mapsto \rho_{ABC},
\end{align}
known as Petz recovery map \cite{Petz1988}.
See Ref. \cite{Fawzi2015} for a construction of a map which is robust to condition \eqref{eq:CMI-zero} holding only approximately.
Likewise, in view of the symmetry of the condition under interchange of $A$ and  $C$, $\rho_{ABC}$ can be reconstructed from $\rho_{BC}$ by a map from $B$ to $AB$. 

In fact, eq.(\ref{eq:CMI-zero}) implies that $B$ has a decomposition as a direct sum of tensor products of Hilbert spaces
\begin{align}
\mathcal{H}_B = \bigoplus_j \mathcal{H}_{B_j} =\bigoplus_j \mathcal{H}_{B_j^L}\otimes \mathcal{H}_{B_j^R},
\end{align}
and that the state of $ABC$ has the block diagonal form
\begin{align}
\rho_{ABC} = \bigoplus_j p_j ~\rho_{A B_j^L}\otimes \rho_{B_j^R C}.
\label{eq:quantum-conditional-independence}
\end{align}
Evidently, we can recover $\rho_{ABC}$ from $\rho_{AB}$ by replacing each $\rho_{B_j^R}$ by $\rho_{B_j^R C}$, without touching the system $A$.  

To apply the Markov condition to our holographic setting, consider a holographic code with no punctures, where the state of the physical boundary is pure. We choose $A$, $B$, $C$ to be three disjoint regions whose union is the complete boundary, namely 
\begin{align}
A = R^c, \quad B = R', \quad C = H , 
%\quad AB = H^c = R'R^c, \quad BC = R = H R',
\end{align}
where $R^c$ denotes the boundary region complementary to $R$.
%, and $H^c$ denotes the complement of $H$
Because the state of the complete boundary is pure, $S(ABC)=0$ and $S(H)=S(H^c)$; therefore the condition eq.(\ref{eq:CMI-zero}) becomes 
\begin{align}\label{eq:markov-hole}
S(AB) + S(BC) = S(B) \iff S(H) + S(R) = S(R').
\end{align}
When this condition is satisfied, $R'$ can be divided into two subsystems, where one purifies the state of $H$ and the other purifies the state of $R^c$. To correct the erasure of $H$ we need only restore the entanglement between $R'$ and $H$, and for this purpose there is no need to venture outside $R$. 

\subsection{Hyperbolic bulk}
Using the Ryu-Takayanagi formula, this statement eq.(\ref{eq:markov-hole}) about entropy would follow from a statement about minimal surfaces:
\begin{equation}
\chi_{R'} = \chi_H \sqcup \chi_R, 
\end{equation}
which is the same as the condition (discussed in \S \ref{sec:uberholography}) for the entanglement wedge $\cE]R']$ to be $\cE[R]\setminus \cE[H]$. For the case where the bulk is a hyperbolic disk, then, the calculation in \S \ref{sec:uberholography} shows that erasure of $H$ can be corrected by a recovery map which acts on region $R$ containing $H$, where $|R|/|H| = (1-r)^{-1} = 3+2\sqrt{2} \approx 5.828 $. Thus the erasure error is locally correctable. This local correctability is a general feature of holographic codes with asymptotically uniform negative bulk curvature. 

We may also consider the case of a manifold with punctures, where the logical boundary $\Lambda$ is maximally entangled with a reference system. In that case the entropy of the physical boundary matches $S(\Lambda)$, and the Markov condition is satisfied provided that 
\begin{align}
|\chi_{R'}| +| \chi_\Lambda | = |\chi_{H^c} |+ |\chi_R|,
\end{align}
which holds if
\begin{align}\label{eq:markov-puncture}
\chi_{R'} = \chi_H \sqcup \chi_R, \quad \chi_{H^c} = \chi_H \sqcup \chi_\Lambda.
\end{align}
As for the case without punctures, eq.(\ref{eq:markov-puncture}) will be satisfied if $H$ is a sufficiently small interval on the physical boundary of the hyperbolic disk, and $R$ is an interval containing $H$, where $|R|$ is larger than $|H|$ by a constant factor. 
The interpretation is the same as before; eq.(\ref{eq:markov-puncture}) implies that $R'$ contains a subsystem which purifies $H$, so there is no need to reach outside of $R$ to recover from the erasure of $H$.

It is also notable that if the Markov condition eq.(\ref{eq:CMI-zero}) is {\it approximately} satisfied, then a local recovery map can be constructed which approximately corrects the erasure of $H$. 
This is important because in realistic AdS/CFT the Markov condition is not exactly satisfied due to the small corrections to the Ryu-Takayanagi formula which we have neglected. 
Local correctability in the approximate setting has been discussed recently in \cite{Flammia2016,Pastawski2016,Kim2016}.

% New
While holographic codes based on tensor networks can successfully reproduce the Ryu-Takayanagi relation satisfied by the Von Neumann entanglement entropy of the boundary theory \cite{Pastawski2015, Hayden2016}, they do not capture correctly the properties of R{\'e}nyi entropies \cite{Headrick2010, dong2016gravity, Dong2016Shape}, and therefore do not provide a fully satisfactory description of conformal field theories with dual geometries. It is fortunate that the Markov condition eq.(\ref{eq:CMI-zero}), and its approximate version \cite{Sutter2016}, are stated in terms of Von Neuman entanglement entropies. We therefore expect that holographic tensor network codes can provide a reasonable picture of local correctability in realistic holography.

\subsection{Flat bulk}
The criterion for local correctability is satisfied by generic negatively curved bulk geometries, but not by bulk geometries which are flat or positively curved. Consider for example a Euclidean two-dimensional disk with unit radius. For an interval $R$ on the boundary which subtends angle $\theta$, the geodesic $\chi_R$ is a chord of the boundary circle with length $|\chi_R| = 2  \sin(\theta/2)$. Suppose we erase a hole $H$ which subtends angle $2\delta$, and wish to correct the erasure by acting in a larger region $R$ that contains $H$. If $R=R_1HR_2$ subtends angle $2\phi$, where $|R_1|=|R_2|$, the Markov condition can be satisfied only if
\begin{align}
&|\chi_{R_1}|+|\chi_{R_2}|= 4 \sin\left((\phi - \delta)/2 \right) \notag\\
\ge& |\chi_R| + |\chi_H| = 2 \sin(\phi) + 2 \sin(\delta).
\end{align}
If $\delta$ and $\phi$ are small, this condition becomes, to leading order in small quantities,
\begin{equation}
\delta \le \phi^3/16.
\end{equation}
Thus when $|H|$ is small, $|R| \sim |H|^{1/3}$ is far larger; the erasure is not locally correctable. The same will be true, even more so, for a positively curved bulk geometry.

The failure of local correctability for holographic codes associated with flat and positively curved bulk manifolds suggests that in these cases the physics of the boundary system is highly nonlocal; in particular, the boundary state is not likely to be the ground state of a local Hamiltonian. This conclusion is reinforced by the observation that, according to the Ryu-Takayanagi formula, the entanglement entropy of a small connected region on the boundary of a flat ball scales linearly with the boundary volume of the region; this  strong violation of the entanglement area law  would not be expected in the ground state if the Hamiltonian is local. 

That flat bulk geometry implies nonlocal boundary physics also teaches a valuable lesson about AdS/CFT. A holographic tensor network provides not just an isometric map from the logical boundary $\Lambda$ to the physical boundary $\Phi$ of the manifold $B$, but also a map from $\Lambda$ to the boundary $\partial X$  of a bulk region $X$  which contains $\Lambda$. If the bulk geometry of $X$ is flat or nearly flat, the entanglement structure of holographic codes indicates that the system supported on $\partial X$ should exhibit flagrant violations of bulk locality. This picture suggests that black holes in the bulk which are small compared to the AdS curvature scale ought to  have highly nonlocal dynamics, as seems necessary for these small black holes to be fast scramblers of quantum information \cite{Hayden2007,Sekino2008}.

\subsection{Positively curved bulk}
This nonlocality of boundary physics is even more pronounced for holographic codes defined on positively curved manifolds. Consider the extreme case of a two-dimensional hemisphere $B$, with the boundary $\partial B$ at its equator. Geodesics on the sphere are great circles, lying  in a plane which passes through the sphere's center. 
For any boundary region $R$ with $|R|\ne |\partial B|/2$, there is a unique minimal surface $\chi_R$, which lies in $\partial B$; for $|R| < |\partial B|/2$ we have $\chi_R = R$, while for $|R| > |\partial B|/2$ we have $\chi_R = \partial B \setminus R$. Invoking the Ryu-Takayanagi formula, we see that as $R$ increases in size, the entropy $S(R)$ rises linearly as $|R|$ until $R$ occupies half of the boundary, and then decreases linearly thereafter. This behavior is the same as for a Haar-random pure state \cite{Page1993}. 

For $|R| < |\partial B|/2$ the entanglement wedge $\cE[R]$ contains only $R$, while for $|R| >|\partial B|/2$, we have $R \sqcup \chi_R = \partial B$, and the entanglement wedge $\cE[R]$ is all of the hemisphere $B$. Accordingly, for any region $X$ in the bulk with associated bulk logical algebra $\cA_X$, the price $p(\cA_X)$ and distance $d(\cA_X)$ are both given by $|\partial B|/2$. This also mimics the behavior of a Haar-random pure state. 

If we regulate bulk and boundary by introducing a lattice spacing $a$, then the number $n$ of boundary sites in the code block is
\begin{align}
n = |\partial B| / a= 2\pi r/a
\end{align}
where $r$ is the radius of the sphere. It is noteworthy that the area of the hemisphere, expressed in lattice units, is
\begin{align}
|B|/a^2 = 2 \pi r^2 / a^2 = n^2/2\pi,
\end{align}
which is quadratic in $n$. 
Since the hemisphere is the surface of smallest area whose boundary reproduces the entanglement structure of a Haar-random state, it is tempting to interpret the area $n^2/2\pi$ as a measure of the circuit complexity of preparing this state, in accord with the {\it complexity equals action conjecture} \cite{Brown2016}. Indeed, a random geometrically local circuit in the bulk containing $O(n^2)$ gates can closely approximate a unitary 2-design, which prepares a state with the desired properties \cite{Brandao2016}. 

As noted in \S \ref{subsec:beyond-finite}, we can crudely model a black hole in the bulk by punching a hole in the bulk, with the microstates of the black hole residing on the logical boundary $\Lambda$ of the puncture. If we introduce a reference system which is maximally entangled with $\Lambda$, then the black hole microstates are maximally mixed, and the entanglement  entropy $S(\Lambda) = |\Lambda|$ counts these microstates, in agreement with the black hole's Bekenstein-Hawking entropy. Alternatively, we might wish to describe a black hole in a typical pure state, rather than a highly mixed state. Our observations about the properties of a holographic code defined on a hemisphere suggest how this can be done. Instead of entangling its boundary with a reference system, we fill the puncture with hemispherical cap. The tensor network filling this cap realizes the minimal geometrically local procedure for preparing the black hole's highly scrambled pure state. 

\section{The holographic strong quantum Singleton bound}\label{sec:holographic-singleton}

In \S \ref{subsec:distance-price} we discussed the strong quantum Singleton bound, Corollary \ref{coro:strong-singleton}, which relates $p$, $d$, and $k$ for a code subspace, and we left open whether this bound can be extended to more general logical operator algebras. Here we will see that for holographic codes such an extension is possible. 

We consider the case of a holographic code with punctures in the bulk; hence there is a physical boundary $\Phi$ and a logical boundary $\Lambda$ as discussed in \S \ref{subsec:beyond-finite}. In the formal continuum limit, the code parameters $p$, $d$, and $k$ are measured in units of area, and the contribution to the area from O(1) boundary sites can be neglected; hence eq.(\ref{eq:strong-singleton}) becomes 
\begin{align}
k \le p - d. 
\end{align}
We would like to show that this constraint applies to logical subalgebras of a holographic code. 

We consider a region $X$ in the bulk, and its associated logical subalgebra $\cA_X$. The region $X$ may contain a portion of the logical boundary $\Lambda$, as well as some additional isolated points in the bulk. We denote the intersection $X\cap\Lambda$ of $X$ with the logical boundary by $\Lambda_X$; if $\Lambda_X$ is nonempty, then the bulk points are a negligible portion of the subalgebra $\cA_X$, whose size is therefore
\begin{align}
k_X = |\Lambda_X|.
\end{align}

In what follows, for the sake of clarity, we will denote the minimal surface associated with boundary region $R$ by $\chi(R)$, in place of the subscript notation $\chi_R$ used earlier. We will also use the notation $R^c$ for the {\it physical} complement $\Phi\setminus R$, and  use $p$, $d$, $k$  as a shorthand for $p(\cA_X)$, $d(\cA_X)$, $k_X$.

Let  $R_p$ be a region of the physical boundary $\Phi$ such that $|R_p| = p(\cA_X)$ and  $\cE[R_p]$ contains $X$; this means that the associated minimal surface $\chi({R_p)}$ must contain $\Lambda_X$. Let $R_d \subseteq R_p$ be a subset of $R_p$ such that, in the regulated theory with a nonzero lattice spacing, $R_d$ contains one less boundary site than the distance of $\cA_X$; therefore $R_d$ is surely correctable with respect to $\cA_X$, and in the continuum limit (where a single site has negligible size) $|R_d| = d(\cA_X)$. Because $R_d$ is correctable,  the entanglement wedge of its physical complement $R_d^c$ contains $X$, which means that $\chi(R_d^c)$ contains $\Lambda_X$.

We may consider gradually ``growing'' a boundary region from $R_d$ to $R_p$, obtaining an inequality by observing that the corresponding minimal surface cannot grow faster than its boundary surface itself:
\begin{align}\label{eq:p-d-chi}
|R_p| - |R_d| \ge \int_{R_d}^{R_p} dR ~\frac{d|\chi(R)|}{dR} = |\chi({R_p})| - |\chi({R_d})|,\notag\\
|R_d^c| - |R_p^c| \ge \int_{R_p^c}^{R_d^c} dR ~ \frac{d|\chi(R)|}{dR} = |\chi({R_d^c})| - |\chi({R_p^c})|. \notag
\end{align}
Together with $p-d = |R_p| - |R_d| = |R_d^c| - |R_p^c|$, this implies
\begin{align}
p-d \ge \frac{|\chi({R_p})| - |\chi({R_d})|+|\chi({R_d^c})| - |\chi({R_p^c})|}{2}.
\end{align} 

Now recall that $\chi({R_p})$ contains $\Lambda_X$. Hence the rest of the minimal surface $\chi(R_p)$, excluding $\Lambda_X$, is $\chi(R_p \cup \Lambda_X)$, or in other words
\begin{align}
\chi(R_p) &=  \chi(R_p \cup \Lambda_X) \cup \Lambda_X  \notag \\
\implies|\chi(R_p)| &=  |\chi(R_p \cup \Lambda_X)| + |\Lambda_X|.
\end{align}
Likewise, $\chi(R_d^c)$ contains $\Lambda_X$, which implies
\begin{align}
|\chi(R_d^c)| =  |\chi(R_d^c \cup \Lambda_X)| + |\Lambda_X|.
\end{align}
Plugging into eq.(\ref{eq:p-d-chi}) yields
\begin{align}\label{eq:p-d-cup}
&p-d - |\Lambda_X| = p-d -k \ge \\
& \frac{1}{2}\left(|\chi({R_p\cup\Lambda_X})|+ |\chi({R_d^c\cup \Lambda_X})|- |\chi({R_d})| - |\chi({R_p^c})|\right).\notag
\end{align}

Now we can use the property that two complementary boundary regions share the same minimal bulk surface (where here by the ``complement'' we mean the boundary complement rather than the physical complement; that is we are simultaneously taking the complement with respect to the logical  and physical  boundaries). 
Let us denote by $\Lambda_X^c$ the complement of $\Lambda_X$ with respect to the {\it logical boundary}, so that $\Lambda=\Lambda_X\Lambda_X^c$. Then,
\begin{align}
\chi(R_d^c \cup \Lambda_X) &= \chi(R_d \cup \Lambda_X^c), \\
\chi(R_p^c) &= \chi(R_p \cup \Lambda_X\Lambda_X^c), 
\end{align}
and hence,
\begin{align}\label{eq:p-d-cup-comp}
p-d -k\ge&|\chi(R_p\cup \Lambda_X)|/2 + |\chi({R_d \cup \Lambda_X^c})|/2 \\
-& |\chi({R_d})|/2 - |\chi(R_p \cup \Lambda_X\Lambda_X^c)|/2.\notag
\end{align}
Using the Ryu-Takayanagi relation between entropy and area, and identifying
\begin{align}
\begin{aligned}
AB &= R_p \cup \Lambda_X,& BC &= R_d \cup \Lambda_X^c, \notag\\  
B&= R_d, & ABC  &= R_p\cup \Lambda_X\Lambda_X^c,
\end{aligned}
\end{align}
the right hand side of eq.(\ref{eq:p-d-cup-comp}) is proportional to 
\begin{align}
S(AB) + S(BC) - S(B) - S(ABC),
\end{align}
which is nonnegative by strong subadditivity of entropy. This completes the holographic proof of the strong quantum Singleton bound:
\begin{theorem}[holographic strong quantum Singleton bound]
Consider a holographic code with logical boundary $\Lambda$, and a logical subalgebra $\cA_X$ associated with bulk region $X$, where $k_X = |X\cap\Lambda|$. Then the price and distance of $\cA_X$ obey
\begin{align}\label{eq:holographic-quantum-singleton}
k_X \le p(\cA_X) - d(\cA_X). 
\end{align}
\end{theorem}

It is intriguing that we used strong subadditivity of entropy in this holographic proof which applies to logical subalgebras, while the proof of Corollary \ref{coro:strong-singleton}, which applies to the price and distance of a traditional code subspace, used only subadditivity. We have not found a proof of the strong quantum Singleton bound that applies to logical subalgebras and does not use holographic reasoning; it is an open question whether eq.(\ref{eq:holographic-quantum-singleton})  holds beyond the setting of holographic codes. 

\section{ Discussion and Outlook}\label{sec:DiscussionANdOutlook}

Our studies of holographic codes have only scratched the surface of a deep subject. There is far more to do, including searches, guided by geometrical intuition, for codes with improved parameters, and investigations of the efficiency of decoding.

Regarding the implications of holographic codes for quantum gravity, we have uncovered several hints which may help to steer future research. We have seen that positive curvature of the bulk manifold can improve properties such as the code distance, but at a cost --- increasing distance is accompanied by enhanced nonlocality of the boundary system. The observation that the logical algebra of a bulk point has price equal to distance is a step toward characterizing bulk geometry using algebraic ideas, and we  anticipate further advances in that direction. Uberholography, in bulk spacetimes with asymptotically negative curvature, illustrates how notions from quantum coding can elucidate the emergence of bulk geometry beyond the appearance of just one extra spatial dimension. 

% New
Following \cite{Almheiri2015} and \cite{Harlow2016}, we have discussed boundary reconstruction of bulk physics using the formalism of operator algebra quantum error correction (OAQEC) \cite{Beny2007,Beny2007a}, which captures salient features of holography. To make firmer contact with realistic AdS/CFT, this discussion should be extended to the setting of {\it approximate} OAQEC \cite{Beny2009}. First steps in this direction have already been taken in \cite{Pastawski2016}, a study of approximate erasure correction in the (1+1)-dimensional Ising CFT at very low temperature, and in \cite{Kim2016}, an investigation of approximate  {\it local} correctability in a MERA network, which has polynomially decaying correlations on its boundary.

We are encouraged by recent progress connecting quantum error correction and quantum gravity, but much remains murky. Most obviously, our discussion of the entanglement wedge and bulk reconstruction applies only to static spacetimes or very special spatial slices through dynamical spacetimes. Applying the principles of quantum coding to more general dynamical spacetimes is an important goal, which poses serious unresolved challenges.

%\subparagraph*{Acknowledgements.}
\acknowledgments
FP would like to thank  Nicolas Delfosse, Henrik Wilming and Jens Eisert for helpful discussions and comments.
FP gratefully acknowledges funding provided by the Institute for Quantum Information and Matter, a NSF Physics Frontiers Center with support of the Gordon and Betty Moore Foundation as well as the Simons foundation through the It from Qubit program and the FUB through ERC project (TAQ).
This research was supported in part by the National Science Foundation under Grant No. NSF PHY-1125915

%%
%% Bibliography
%%

%% Either use bibtex (recommended), 
%\bibliography{QECCpropertiesFromAdSCFT}

%% .. or use the thebibliography environment explicitly (i.e importing the generated bbl file)

%%  APPENDIX
\appendix
\section{Comparison of uberholography dimensional exponent with explicit codes}\label{app:Comparison}

%New
In \S \ref{sec:uberholography} we computed the universal fractal dimension  $\alpha \approx 0.786$ for a holographic code defined on the Poincar\'e disk, assuming that geometric complementarity and the entanglement wedge hypothesis are precisely satisfied. We may also define a fractal dimension for other code families, such as concatenated quantum codes or holographic tensor network codes.  Choosing a particular logical algebra ${\cal A}$ (for example, the algebra of a logical qubit at the center of the bulk), let 
\begin{align}
\alpha_p := \lim_{\ell \rightarrow \infty} \frac{\log{p_\ell}}{\log{n_\ell}}, \quad
\alpha_d := \lim_{\ell \rightarrow \infty} \frac{\log{d_\ell}}{\log{n_\ell}},
\end{align}
where $p_\ell$ and $d_\ell$ are the price and distance of the  logical algebra after $\ell$ levels of concatenation or equivalent iteration. (If ${\cal A}$ is the algebra of local operators at the center of the bulk, then we may think of $\ell$ as the radial distance from the center of the bulk to its boundary.)
The no free lunch lemma ensures that  $\alpha_p \geq \alpha_d$.

% concatenated code
By a concatenated code we mean a recursive hierarchy of codes within codes; these can be constructed in many ways. In the simplest case, we consider an $ [[n,1,d]]$ code ${\cal C}_1$, with just one logical qubit, which has an encoder isometrically mapping one qubit to a block of $n$ physical qubits.  The code ${\cal C}_2$, which has $k=1$ and length $n_2=n^2$, is obtained by applying this encoder to each of the $n$ physical qubits in ${\cal C}_1$; likewise the code ${\cal C}_{\ell}$, with length $n^{\ell}$, is obtained by applying the encoder to each of the $n^{\ell -1}$ qubits in the code ${\cal C}_{\ell -1 }$. The corresponding tensor network, with one logical qubit at its center, is a branching tree extending radially outward in which each branch has $n$ descendants. 

Suppose that $n$ is odd and the code ${\cal C}_1$ has the largest possible distance $d = (n+1)/2$. The complementarity bound eq.(\ref{eq:p+d}) then implies that the price is $p=d$; the full logical algebra can be supported on $d$ of the $n$ qubits, and all nontrivial logical Pauli operators have weight $d$. Therefore, all nontrivial logical operators of ${\cal C}_{\ell}$ can be supported on $d^{\ell}$ qubits and all have weight $d^{\ell}$. We conclude that
\begin{align}
\alpha_p = \alpha_d = \frac{\log d}{\log n} . 
\end{align}
As $n$ increases, $\alpha_p$ and $\alpha_d$ approach $1$ from below. Although the tensor network can be embedded in a plane, it does not approximate the geometry of the Poincar\'e disk, and its price and distance obey a different scaling law than we found in \S \ref{sec:uberholography}.

A more complicated recursive encoding scheme, based on an $[[n,k,d]]$ code ${\cal C}_1$ with $k > 1$, is depicted in figure \ref{fig:BranchingConcatenation}. In this case, the ${\cal C}_1$ encoder maps $k$ qubits to a block of $n$ qubits. To build the code ${\cal C}_{\ell}$, we first assemble $k$ copies of the code ${\cal C}_{\ell - 1}$, and then apply the $k \to n$  encoder for ${\cal C}_1$ all together $n^{\ell -1 }$ times, where each encoder acts on $k$ qubits drawn from the $k$ distinct copies. Each time we add another layer to the code, the number of encoded qubits increases by a factor of $k$ and the number of physical qubits increases by a factor of $n$, therefore
\begin{align}\label{eq:n-and-k-concatenated}
n_\ell = n^\ell, \quad k_\ell = k^\ell.
\end{align}
However, in this case the price and distance of ${\cal C}_{\ell}$ are not so easy to calculate, though we can derive some simple bounds. When we add an additional layer to the code, each nontrivial logical operator of $C_{\ell-1}$ maps to a logical operator of ${\cal C}_\ell$ whose weight is {\it at least} $d$ times larger; furthermore, if all logical operators of ${\cal C}_{\ell - 1}$ can be supported on $w$ physical qubits, then {\it at most} $pw$ qubits are needed to support all logical operators of ${\cal C}_{\ell}$. We therefore have
\begin{align}\label{eq:concatenated-price-distance}
d^\ell \le d_\ell \le p_\ell\le p^\ell. 
\end{align}
But to make a more precise statement about the price and distance of  $C_\ell$ we need more information about the structure of ${\cal C}_1$.

\begin{figure}[ht]
	\begin{center}
		\includegraphics[width=0.8\columnwidth]{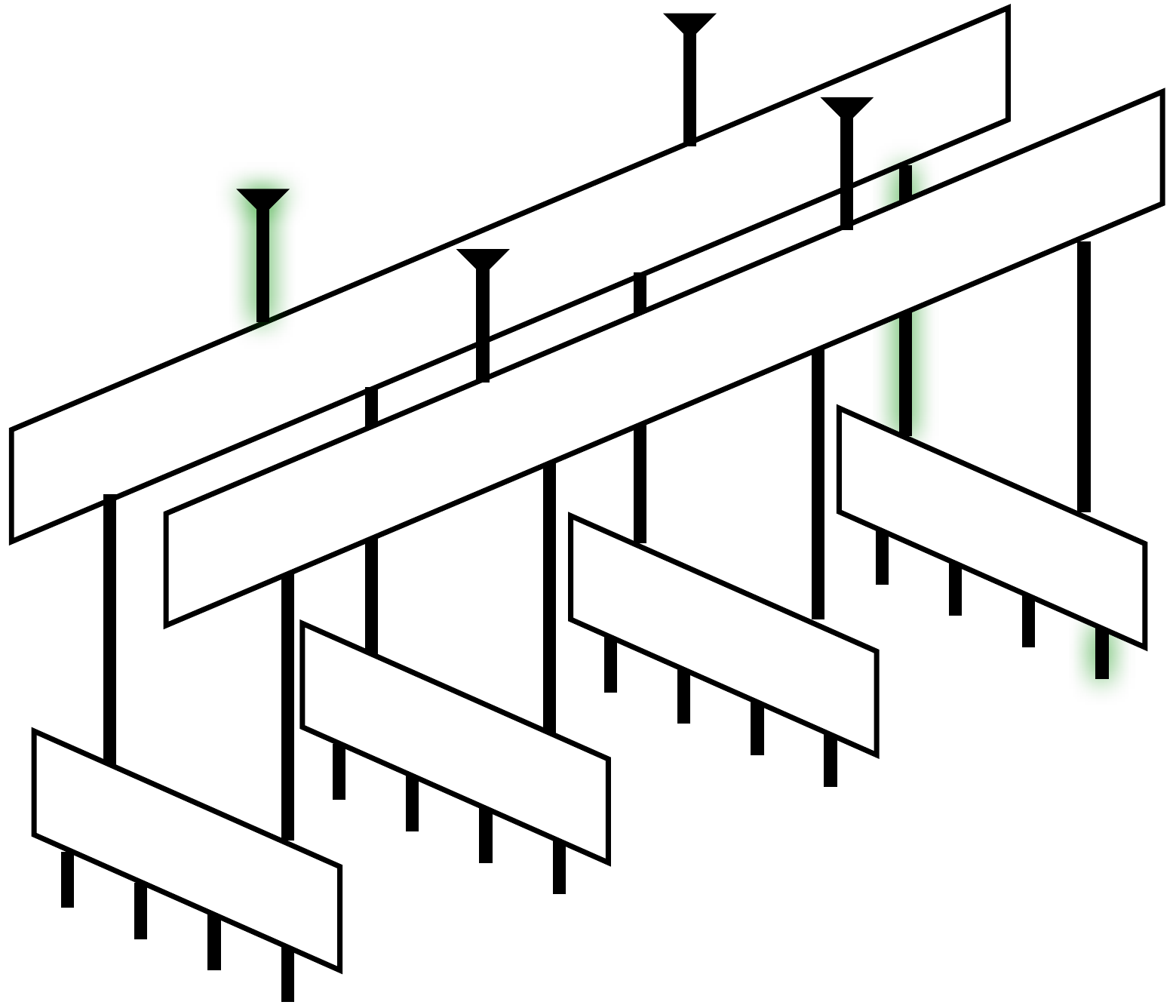}
		\caption{This figure illustrates a recursive coding network to which eq.(\ref{eq:n-and-k-concatenated}) and eq.(\ref{eq:concatenated-price-distance}) apply, with logical qubits at the top and physical qubits at the bottom. To derive that the distance is bounded below by $d^\ell$ and the price is bounded above by $p^\ell$, it suffices to observe that there is a unique ``causal'' path connecting each physical qubit to each logical qubit. Here a $[[4,2,2]]$ code (with encoder drawn as a parallelogram) is concatenated to obtain a $[[16,4,4]]$ code, and a causal path connecting a physical qubit to a logical qubit is highlighted. 
		}\label{fig:BranchingConcatenation}
	\end{center}
\end{figure}

%Holographic pentagon code.
We may also consider the price and distance of holographic tensor network codes, which capture some of the features of full blown AdS/CFT duality \cite{Pastawski2015}. For example, we can tile the Poincar\'e disk with pentagons, and associate a 6-index ``perfect tensor'' with each pentagon, where each pentagon carries a single logical qubit. For this pentagon code, where $\cA$ is the logical algebra of the central pentagon, we find that the price  $p_\ell(\cA)$ and distance $d_\ell(\cA)$ are badly mismatched (where $\ell$ denotes the graph distance form the central pentagon to the physical boundary of the tensor network). In fact, the distance $d_\ell(\cA)=4$ is a constant independent of $\ell$, as explained in \S 5.6 of \cite{Pastawski2015}; in other words, there are logical operators of weight 4 which act nontrivially on the central qubit.
In contrast, we expect the price to scale as $p_\ell (\cA) = n_\ell^{\alpha_{p5}}$ with $0.786 < \alpha_{p5}< 1$ ({\it i.e.},  with an exponent larger than the value attained in idealized holography). This upper bound may be derived using a discrete version of the hole punching approach, thereby constructing explicitly a subset of the physical boundary qubits for which the greedy algorithm of \cite{Pastawski2015} reaches the central tensor. 
It is more difficult to obtain lower bounds on $p_\ell$, as these cannot be witnessed by examples.

A nontrivial scaling exponent for the distance $d_\ell(\cA)$ can be obtained if we thin out the logical qubits, replacing pentagons in the bulk by hexagons which carry no logical qubit index. 
(In such codes the central qubit is well protected against erasure of a randomly chosen nonzero fraction of all the physical boundary qubits, as shown in \cite{Pastawski2015}.) Codes with relatively sparse bulk logical qubits are better suited than the pentagon code for illustrating the ideas we have explored in this paper, where we have focused on the regime in which geometric entanglement dominates bulk entanglement. We may anticipate that holographic tensor network code families which mimic the geometry of the Poincar\'e disk will have a price scaling exponent $\alpha_p$ which approximates $\alpha\approx 0.786$ from above and a distance scaling exponent $\alpha_d$ which approximates $\alpha$ from below. We have confirmed this expectation by studying some examples, though we have no rigorous general argument.

\end{document}